\documentclass[lettersize,journal]{IEEEtran}
\usepackage{amsmath,amsfonts}
\usepackage{algorithmic}
\usepackage{array}
\usepackage[caption=false,font=normalsize,labelfont=sf,textfont=sf]{subfig}
\usepackage{textcomp}
\usepackage{stfloats}
\usepackage{url}
\usepackage{verbatim}
\usepackage{graphicx}
\usepackage{amssymb}
\usepackage{amsthm}
\theoremstyle{plain}
\newtheorem{theorem}{Theorem}
\theoremstyle{plain}
\newtheorem{Lemma}{Lemma}
\theoremstyle{plain}
\newtheorem{prop}{Proposition}
\theoremstyle{plain}
\newtheorem{rem}{Remark}
\def\BibTeX{{\rm B\kern-.05em{\sc i\kern-.025em b}\kern-.08em
    T\kern-.1667em\lower.7ex\hbox{E}\kern-.125emX}}
\usepackage{balance}
\makeatletter
\renewcommand{\maketag@@@}[1]{\hbox{\m@th\normalsize\normalfont#1}}%
\makeatother

\begin{document}
\title{AC-Feasible Power Transfer Regions of Virtual
Power Plants: Characterization and Application}
\author{Wei Lin, \textit{Member}, \textit{IEEE}, Changhong Zhao, \textit{Senior Member}, \textit{IEEE}
\thanks{W. Lin and C. Zhao are with the Department of Information Engineering, the Chinese University of Hong Kong, New Territories, Hong Kong SAR. Emails: \{wlin, chzhao\}@ie.cuhk.edu.hk. (Corresponding author: C. Zhao)

This work was supported by Research Grants Council of Hong Kong through ECS Award No. 24210220 and the CUHK faculty startup grant.}}

\markboth{Journal of \LaTeX\ Class Files,~Vol.~**, No.~**, ****}%
{How to Use the IEEEtran \LaTeX \ Templates}

\maketitle
\vspace{-0.2in}
\begin{abstract}
Distributed energy resources (DERs) in distribution networks can be aggregated as a virtual power plant (VPP) for transmission-level operations. A critical challenge for such coordination is the complexity of the AC-feasible power transfer region between a VPP and the transmission system at their point of common coupling. To overcome this challenge, this paper develops a characterization method for such regions. The proposed method constructs linear constraints to inner-approximate the AC-feasible power transfer regions. To guarantee AC-feasibility, the parameters in these constraints are determined by applying the Brouwer fixed point theorem to the second-order Taylor expansion of the nonlinear Dist-Flow equations. Based on the power transfer regions characterized with our method, a transmission-level operation problem with VPP participation is formulated and solved through big-M linearization. The proposed methods are verified by numerical experiments in the IEEE 33-bus and IEEE 136-bus test systems.
\end{abstract}


\section{Introduction}
\IEEEPARstart{T}{he} percentages of \textit{distributed energy resources} (DERs) in distribution networks are increasing in power supplies \cite{1}. For transmission-level operations with the participation of DERs, the concept of a \textit{virtual power plant} (VPP) has drawn much attention, leading to practical projects (e.g., FENNIX \cite{2} and EDISON \cite{3}). A VPP paves one promising way to aggregate an entire distribution network with its DERs as a participant in transmission-level operations by adjusting its power transfers at the \textit{point of common coupling} (PCC) \cite{4}. Due to different operators of VPPs and the transmission network, a promising method for their coordination is based on characterizing the feasible power transfer region of a VPP \cite{5}. This region is a polytope in the domain of PCC power transfers which can be physically executed while respecting the VPP’s operational constraints. Such a region can be applied as constraints in transmission-level operations.

The current methods to characterize feasible power transfer regions can be categorized based on the power flow model used therein.

The first type of methods is based on linearized power flow models, mainly including  multi-parametric linear programming methods \cite{6}-\cite{7}, the vertex search method \cite{8}, the Fourier-Motzkin elimination method \cite{9}, the Stackelberg-gam-based method \cite{10}, and the robust-optimization-based methods \cite{11}-\cite{12}. The linearized power flow models used above have been observed to facilitate the constructions of feasible regions, e.g., through the easy matrix multiplication in linear systems in \cite{6}-\cite{9} and the strong duality of optimization problems in \cite{10}-\cite{12}. However, the linearized power flow models therein are approximations to the full AC power flow model with approximation errors; thus, their resulting feasible regions may not guarantee AC-feasibility.

In contrast, the full AC power flow model was directly adopted in the second type of methods. These methods can be further discussed based on what constraints are considered and guaranteed. Firstly, Banach fixed-point theorem in \cite{13}-\cite{14}, Brouwer fixed-point theorem in \cite{15}-\cite{16}, and Kantorovich fixed-point theorem in \cite{17} were employed to construct feasible regions to guarantee power flow solvability, while safety limits (e.g., voltage and current limits) were neglected. With the further consideration of safety limits, a feasible region can be calculated by optimization methods in \cite{18}-\cite{19}. Ref. \cite{18} solved non-convex optimizations to find boundary points which serve as a convex-hull-based approximation to the true feasible region, while Ref. \cite{19} reported a heuristic approach to remove infeasible regions based on the SOCP relaxation.

Different with \cite{18}-\cite{19} which do not guarantee power flow solvability and safety limits, Refs. \cite{20}-\cite{26} provide feasible regions with such guarantees. Refs. \cite{20}-\cite{21} characterized feasible regions by identifying the regular stable equilibrium manifolds of a quotient gradient system. However, Refs. \cite{20}-\cite{21} could only numerically search for the points, without providing an explicit formulation of the feasible region. This absence of explicit characterization may limit the application of the feasible region in transmission-level operations. Refs. \cite{22}-\cite{24} provide feasible regions based on tightened AC power flow formulations, while their power flow equations and safety limits are guaranteed under certain special conditions (e.g., the special monotonicity of the functions of proxy variables in \cite{22}-\cite{23}, and a special algebraic relationship of network parameters in \cite{24}). Ref. \cite{25} formulated and solved a nonconvex optimization based on the self-mapping of a polytope to get a sub-region whose shape is conservatively assumed as a box. Ref. \cite{26} further derived quadratic constraints to delineate a feasible region without a pre-set shape. However, linear constraints are current industry preferences in operations, planning and many practical applications \cite{27}-\cite{28}. This viewpoint motivates us in this paper to characterize a feasible region based on linear constraints. The major efforts in this paper are summarized below.

(1) A characterization method for the AC-feasible power transfer region of a VPP is proposed (Sec. III). Given a point in the domain of PCC power transfers, a mixed-integer nonlinear program (MINLP) is developed to get linear constraints for an AC-feasible sub-region around that point, by combining the Brouwer fixed point theorem with the second-order Taylor expansion of the nonlinear Dist-Flow equations. The proposed MINLP can be solved through decomposition into smaller-scale NLPs across the network branches. An exploration strategy is further discussed to find more AC-feasible sub-regions based on the vertices of the previously found sub-regions. The union of the found sub-regions can serve as an inner approximation to the true AC-feasible power transfer region.

(2) An application method for the AC-feasible power transfer region of a VPP is proposed (Sec. IV). Based on the linear constraints to characterize multiple AC-feasible sub-regions obtained above, a big-M formulation is developed to linearize the transmission-level operation problem with VPP participation. 

In Sec. V, the proposed methods are numerically validated in the IEEE 33-bus and the IEEE 136-bus test systems. Conclusions of this paper are summarized in Sec. VI

\section{Operational Constraints of VPPs}
 Consider a VPP encompassing a radial distribution network. Let $N=\left\{1, \ldots, n_{\mathrm{N}}\right\}$ denote the set of nodes. Let $G=\left\{1, \ldots, n_{\mathrm{G}}\right\}$. Let $L=\left\{1, \ldots, n_{\mathrm{L}}\right\}$ denote the set of branches. Particularly, the branches are treated as directed, i.e., the branch that connects nodes $i, j \in N$ is denoted by $i \longrightarrow j$ where node $i$ is closer to the root node than node $j$. The Dist-Flow equations in \cite{29} are employed in this paper to model the operational constraints of a VPP, as below.
\begin{equation}\label{D1}
\begin{aligned}
&P_{i j}-R_{i j} I_{i j}+\sum_{g \in G} e_{j g} P_{g}+e_{j 0} P_{\mathrm{PCC}} \\
&=\sum_{k, j \rightarrow k} P_{j k}+\sum_{n \in N} e_{j n} P_{n}, \forall j \in N,
\end{aligned}
\end{equation}
\begin{equation}\label{D2}
\begin{aligned}
&Q_{i j}-X_{i j} I_{i j}+\sum_{g \in G} e_{j g} Q_{g}+e_{j 0} Q_{\mathrm{PCC}} \\
&=\sum_{k, j \rightarrow k} Q_{j k}+\sum_{n \in N} e_{j n} Q_{n}, \forall j \in N,
\end{aligned}
\end{equation}
\begin{equation}\label{D3}
\begin{aligned}
V_{j} &=V_{i}-2\left(R_{i j} P_{i j}+X_{i j} Q_{i j}\right) \\
&+\left(\left(R_{i j}\right)^{2}+\left(X_{i j}\right)^{2}\right) I_{i j}, \forall i \rightarrow j,
\end{aligned}
\end{equation}
\begin{equation}\label{D4}
V_{i} I_{i j}=\left(P_{i j}\right)^{2}+\left(Q_{i j}\right)^{2}, \forall i \rightarrow j,
\end{equation}
\begin{equation}\label{D5}
V_{i}^{\min } \leq V_{i} \leq V_{i}^{\max }, \forall i \in N,
\end{equation}
\begin{equation}\label{D6}
I_{i j}^{\min } \leq I_{i j} \leq I_{i j}^{\max }, \forall i \rightarrow j,
\end{equation}
\begin{equation}\label{D7}
P_{g}^{\min } \leq P_{g} \leq P_{g}^{\max }, \forall g \in G,
\end{equation}
\begin{equation}\label{D8}
Q_{g}^{\min } \leq Q_{g} \leq Q_{g}^{\max }, \forall g \in G,
\end{equation}
\begin{equation}\label{D9}
P_{n}^{\min } \leq P_{n} \leq P_{n}^{\max }, \forall n \in N,
\end{equation}
\begin{equation}\label{D10}
Q_{n}^{\min } \leq Q_{n} \leq Q_{n}^{\max }, \forall n \in N,
\end{equation}
where $P_{i j}$ and $Q_{i j}$ are active and reactive branch power flows from node $i$ to node $j$, respectively; $I_{i j}$ is the squared current magnitude from node $i$ to node $j$; $P_g$ and $Q_g$ are active and reactive generation levels of unit $g$, respectively; $P_{\mathrm{PCC}}$ and $Q_{\mathrm{PCC}}$ are active and reactive power transfers at the PCC, respectively; $P_n$ and $Q_n$ are active and reactive power demand at node $n$, respectively. $R_{i j}$ and $X_{i j}$ are constant resistance and reactance from node $i$ to node $j$; $e_{j g}$, $e_{j 0}$, and $e_{j n}$ are incident indicators; $V_{i}$ is the squared voltage magnitude at node $i$; the superscripts “max” and “min” indicate upper and lower bounds, respectively.

For convenience of discussion, $\left(P_{i j}, Q_{i j}\right)$ of all the branches is stacked as $\left(\mathbf{P}_{\mathrm{L}}, \mathbf{Q}_{\mathrm{L}}\right)$ where $\mathbf{P}_{\mathrm{L}} \in \mathbb{R}^{n_{\mathrm{L}} \times 1}$ and $\mathbf{Q}_{\mathrm{L}} \in \mathbb{R}^{n_{\mathrm{L}} \times 1}$, $\left(P_{g}, Q_{g}\right)$ of all the distributed generation units is stacked as $\left(\mathbf{P}_{\mathrm{G}}, \mathbf{Q}_{\mathrm{G}}\right)$ where $\mathbf{P}_{\mathrm{G}} \in \mathbb{R}^{n_{\mathrm{G}} \times 1}$ and $\mathbf{Q}_{\mathrm{G}} \in \mathbb{R}^{n_{\mathrm{G}} \times 1}$, $\left(P_{n}, Q_{n}\right)$ of all the demand is stacked as $\left(\mathbf{P}_{N}, \mathbf{Q}_{\mathrm{N}}\right)$ where $\mathbf{P}_{\mathrm{N}} \in \mathbb{R}^{n_{\mathrm{N}} \times 1}$ and $\mathbf{Q}_{\mathrm{N}} \in \mathbb{R}^{n_{\mathrm{N}} \times 1}$, $V_{i}$ of all the nodes is stacked as $\mathbf{V} \in \mathbb{R}^{n_{\mathrm{N}} \times 1}$, $I_{i j}$ of all the branches is stacked as $\mathbf{I} \in \mathbb{R}^{n_{\mathrm{L}} \times 1}$ and power transfers $\left(P_{\mathrm{PCC}}, Q_{\mathrm{PCC}}\right)$ at the PCC is stacked as $\mathbf{u}_{\mathrm{PCC}} \in \mathbb{R}^{n _{\mathrm{PCC}} \times 1}$ where $n_{\mathrm{PCC}}=2$.

The controllable variables in a VPP are stacked as a vector $\mathbf{u}=\left[\begin{array}{llll}\mathbf{P}_{\mathrm{G}}^{T} & \mathbf{Q}_{\mathrm{G}}^{T} & \mathbf{P}_{\mathrm{N}}^{T} & \mathbf{Q}_{\mathrm{N}}^{T}\end{array}\right]^{T} \in \mathbb{R}^{n_{\mathbf{u}} \times 1}$ where $n_{\mathrm{u}}=2 n_{\mathrm{G}}+2 n_{\mathrm{N}}$. The state variables include $\mathbf{P}_{\mathrm{L},} \mathbf{Q}_{\mathrm{L},} \mathbf{V}$, and $\mathbf{I}$ whose total dimension is $n_{\mathrm{x}}=3 n_{\mathrm{L}}+n_{\mathrm{N}}$. For convenience of discussion, the compact formulation of the constraints (\ref{D1})-(\ref{D10}) is given below.
\begin{equation}\label{C1}
\mathbf{f}\left(\mathbf{P}_{\mathrm{L}}, \mathbf{Q}_{\mathrm{L}}, \mathbf{V}, \mathbf{I}\right)=\mathbf{K}_{1} \mathbf{u}_{\mathrm{PCC}}+\mathbf{K}_{2} \mathbf{u},
\end{equation}
\begin{equation}\label{C2}
\mathbf{V}_{\min } \leq \mathbf{V} \leq \mathbf{V}_{\max },
\end{equation}
\begin{equation}\label{C3}
\mathbf{I}_{\min } \leq \mathbf{I} \leq \mathbf{I}_{\max },
\end{equation}
\begin{equation}\label{C4}
\mathbf{u}_{\min } \leq \mathbf{u} \leq \mathbf{u}_{\max },
\end{equation}
where $\mathbf{f}=\left[\begin{array}{llll}f_{1} & f_{2} & \ldots & f_{n_{\mathrm{e}}}\end{array}\right]^{T}: \mathbb{R}^{n_{\mathrm{x}} \times 1} \rightarrow \mathbb{R}^{n_{\mathrm{e}} \times 1}$ and $n_{\mathrm{e}}=2 n_{\mathrm{L}}+2 n_{\mathrm{N}}$; $\mathbf{K}_{1} \in \mathbb{R}^{n_{\mathrm{e}} \times n_{\mathrm{PCC}}}$; $\mathbf{I}_{\min }$ and $\mathbf{I}_{\max } \in \mathbb{R}^{n_{\mathrm{L}} \times 1}$; $\mathbf{u}_{\min }$ and $\mathbf{u}_{\max } \in \mathbb{R}^{n_{\mathrm{u}} \times 1}$. The power flow equations (\ref{D1})-(\ref{D4}) are compacted in (\ref{C1}), the safety limits (\ref{D5})-(\ref{D6}) are compacted in  (\ref{C2})-(\ref{C3}), and the limits (\ref{D7})-(\ref{D10}) of controllable variables are compacted in (\ref{C4}).

The coordination between a VPP and a transmission network lies in $\mathbf{u}_{\mathrm{PCC}}$. To participate in  transmission-level operations, the operator of a VPP needs to characterize an AC-feasible transfer region in the domain of $\mathbf{u}_{\mathrm{PCC}}$ with which there is $\left(\mathbf{u}, \mathbf{P}_{\mathrm{L}}, \mathbf{Q}_{\mathrm{L}}, \mathbf{V}, \mathbf{I}\right)$ that satisfies (\ref{C1})-(\ref{C4}). Such a point $\mathbf{u}_{\mathrm{PCC}}$ is called AC-feasible, and the set of all AC-feasible points in the domain of $\mathbf{u}_{\mathrm{PCC}}$ is defined as the AC-feasible power transfer region:
\begin{equation}
\boldsymbol{\Omega} \triangleq\left\{\begin{array}{c|c}
\mathbf{u}_{\mathrm{PCC}} & \exists\left(\mathbf{u}_{\mathrm{PCC}}, \mathbf{u}, \mathbf{P}_{\mathrm{L}}, \mathbf{Q}_{\mathrm{L}}, \mathbf{V}, \mathbf{I}\right) \\
\in \mathbb{R}^{n_{\mathrm{PCC}} \times 1} &\textnormal{that satisfies}~(11)-(14)
\end{array}\right\}.
\end{equation}

In coming Sec. III, we will present a method to characterize the AC-feasible power transfer region $\boldsymbol{\Omega}$ with linear constraints.

\section{Characterization of AC-feasible Power Transfer Regions}

 The key idea of our characterization method is given in Fig. \ref{Fig.1}. Our method iteratively explores the sub-regions of the AC-feasible power transfer region $\boldsymbol{\Omega}$. In each iteration, the linear constraints will be constructed to characterize sub-regions in which every $\mathbf{u}_{\mathrm{PCC}}$ guarantees AC-feasibility around the current search point (e.g., the black dot in Fig. \ref{Fig.1}). The method to construct such a sub-region will be elaborated in Sec. III-A. Furthermore, an exploration strategy will be discussed to find more sub-regions based on the previously found sub-regions, as will be elaborated in Sec. III-B. The main idea of this strategy is that the vertices of the found sub-regions will be selected as the candidates for the next search points (e.g., the blue triangles in Fig. \ref{Fig.1}). The union of all the found sub-regions can serve as an inner approximation to the AC-feasible power transfer region $\boldsymbol{\Omega}$.
 
\begin{figure}
\centering
\setlength{\belowcaptionskip}{-6cm} 
\includegraphics[width=3.2in]{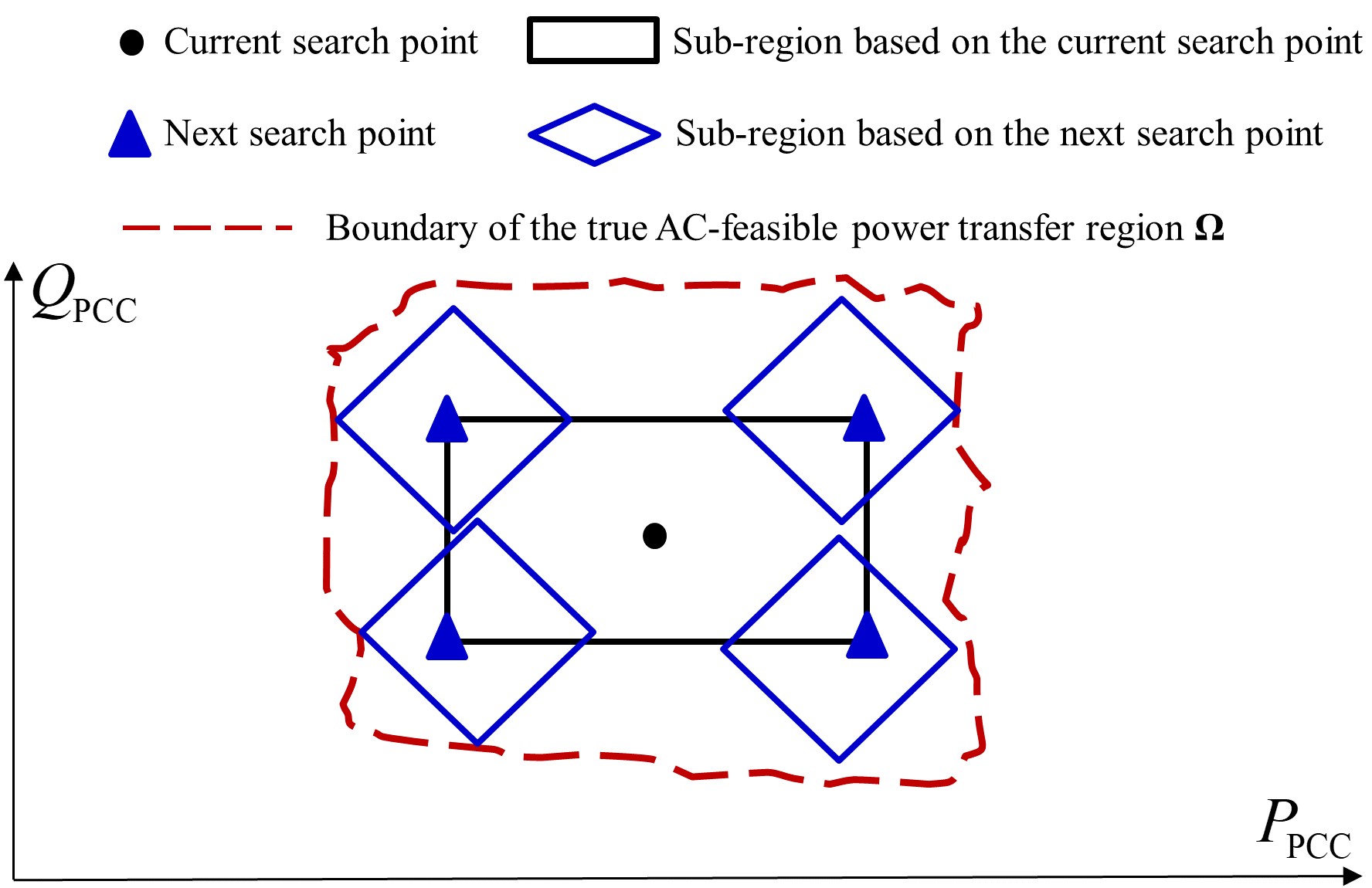}
\caption{An illustration of the proposed characterization method.}
\label{Fig.1}
\end{figure}

\subsection{Constructing an AC-feasible sub-region}
Our construction stems from a fixed-point representation of the power flow equation (\ref{C1}) around a known $\mathbf{u}_{\mathrm{PCC0}}$. Furthermore, the Brouwer’s fixed point theorem \cite{30} is employed to analyze the requirements to guarantee the solvability and safety of the fixed-point representation. Finally, the resulting requirements are used to formulate the linear constraints to characterize a sub-region of the AC-feasible power transfer region in the domain of $\mathbf{u}_{\mathrm{PCC}}$.

\subsubsection{Fixed-point representation of power flow}

Suppose given $\mathbf{u}_{\mathrm{PCC0}}$, a known $\left(\mathbf{u}_{0}, \mathbf{P}_{\mathrm{L} 0}, \mathbf{Q}_{\mathrm{L} 0}, \mathbf{V}_{0}, \mathbf{I}_{0}\right)$ exists with respect to the constraints (\ref{C1})-(\ref{C4}). For convenience of discussion, $\left(\mathbf{P}_{\mathrm{L}}, \mathbf{Q}_{\mathrm{L}}, \mathbf{V}, \mathbf{I}\right)$ is stacked as $\mathbf{X}$, and $\mathbf{x}_{0}=\left[\begin{array}{llll}\mathbf{P}_{\mathrm{L} 0}^{T} & \mathbf{Q}_{\mathrm{L} 0}^{T} & \mathbf{V}_{0}^{T} & \mathbf{I}_{0}^{T}\end{array}\right]^{T}$. For the $i^{t h}$ element $f_{i}(\mathbf{x})$ of $\mathbf{f}(\mathbf{x})$ in (\ref{C1}),   its second-order Taylor expansion is exact because $f_{i}(\mathbf{x})$ is linear or quadratic, i.e.,
\begin{equation}\label{F1}
\begin{aligned}
f_{i}(\mathbf{x}) &=f_{i}\left(\mathbf{x}_{0}\right)+\left.\nabla_{\mathbf{x}} f_{i}(\mathbf{x})\right|_{\mathbf{x}=\mathbf{x}_{0}}\left(\mathbf{x}-\mathbf{x}_{0}\right) \\
&+\left(\mathbf{x}-\mathbf{x}_{0}\right)^{T} \frac{\left.\nabla_{\mathbf{x}}^{2} f_{i}(\mathbf{x})\right|_{\mathbf{x}=\mathbf{x}_{0}}}{2}\left(\mathbf{x}-\mathbf{x}_{0}\right).
\end{aligned}
\end{equation}
where $\left.\nabla_{\mathbf{x}} f_{i}(\mathbf{x})\right|_{\mathbf{x}=\mathbf{x}_{0}} \in \mathbb{R}^{1 \times n_{\mathbf{x}}}$ and $\left.\nabla_{\mathbf{x}}^{2} f_{i}(\mathbf{x})\right|_{\mathbf{x}=\mathbf{x}_{0}} \in \mathbb{R}^{n_{\mathbf{x}} \times n_{\mathbf{x}}}$ can be determined at a given point $\mathbf{x}_{0}$.

Stacking (\ref{F1}) for each $f_{i}(\mathbf{x})$ yields an equivalent formulation of $\mathbf{f}(\mathbf{x})$ in (\ref{C1}), i.e., $\mathbf{f}(\mathbf{x})$ in (\ref{C1}) can be re-written as
\begin{equation}\label{F2}
\mathbf{f}(\mathbf{x})=\mathbf{f}\left(\mathbf{x}_{0}\right)+\left.\nabla_{\mathbf{x}} \mathbf{f}(\mathbf{x})\right|_{\mathbf{x}=\mathbf{x}_{0}}\left(\mathbf{x}-\mathbf{x}_{0}\right)+\mathbf{R},
\end{equation}
where $\left.\nabla_{\mathbf{x}} \mathbf{f}(\mathbf{x})\right|_{\mathbf{x}=\mathbf{x}_{0}} \in \mathbb{R}^{n_{\mathrm{e}} \times n_{\mathbf{x}}}$ is assumed to be non-singular at $\mathbf{x}_{0}$; $\mathbf{R} \in \mathbb{R}^{n_{e^{\times 1}}}$ is the second-order term and the $i^{t h}$ element $R_{i}$ in $\mathbf{R}$ is denoted as $\left.\left(\mathbf{x}-\mathbf{x}_{0}\right)^{T} \nabla_{\mathbf{x}}^{2} f_{i}(\mathbf{x})\right|_{\mathbf{x}=\mathbf{x}_{0}}\left(\mathbf{x}-\mathbf{x}_{0}\right)$.

Furthermore, replacing $\mathbf{f}(\mathbf{x})$ in (\ref{C1}) with (\ref{F2}) yields
\begin{equation}\label{F3}
\mathbf{K}_{1} \mathbf{u}_{\mathrm{PCC}}+\mathbf{K}_{2} \mathbf{u}=\mathbf{f}\left(\mathbf{x}_{0}\right)+\left.\nabla_{\mathbf{x}} \mathbf{f}(\mathbf{x})\right|_{\mathbf{x}=\mathbf{x}_{0}}\left(\mathbf{x}-\mathbf{x}_{0}\right)+\mathbf{R}.
\end{equation}

For $\left(\mathbf{x}_{0}, \mathbf{u}_{\mathrm{PCC} 0,}, \mathbf{u}_{0}\right)$, (\ref{C1}) holds as $\mathbf{f}\left(\mathbf{x}_{0}\right)=\mathbf{K}_{1} \mathbf{u}_{\mathrm{PCC} 0}+\mathbf{K}_{2} \mathbf{u}_{0}$. Consequently, substituting $\mathbf{f}\left(\mathbf{x}_{0}\right)=\mathbf{K}_{1} \mathbf{u}_{\mathrm{PCC} 0}+\mathbf{K}_{2} \mathbf{u}_{0}$ into (\ref{F3}) yields
\begin{equation}\label{F4}
\begin{aligned}
&\mathbf{K}_{1}\left(\mathbf{u}_{\mathrm{PCC}}-\mathbf{u}_{\mathrm{PCCO}}\right)+\mathbf{K}_{2}\left(\mathbf{u}-\mathbf{u}_{0}\right) \\
&=\left.\nabla_{\mathbf{x}} \mathbf{f}(\mathbf{x})\right|_{\mathbf{x}=\mathbf{x}_{0}}\left(\mathbf{x}-\mathbf{x}_{0}\right)+\mathbf{R} \\
&\Rightarrow \mathbf{K}_{1} \widetilde{\mathbf{u}}_{\mathrm{PCC}}+\mathbf{K}_{2} \widetilde{\mathbf{u}}=\left.\nabla_{\mathbf{x}} \mathbf{f}(\mathbf{x})\right|_{\mathbf{x}=\mathbf{x}_{0}} \widetilde{\mathbf{x}}+\mathbf{R}
\end{aligned}
\end{equation}
where $\widetilde{\mathbf{u}}_{\mathrm{PCC}}=\left(\mathbf{u}_{\mathrm{PCC}}-\mathbf{u}_{\mathrm{PCC} 0}\right) \in \mathbb{R}^{n_{\mathrm{PCC}} \times 1}$; $\widetilde{\mathbf{u}}=\left(\mathbf{u}-\mathbf{u}_{0}\right) \in \mathbb{R}^{n_{\mathrm{u}} \times 1}$; $\widetilde{\mathbf{x}}=\left(\mathbf{x}-\mathbf{x}_{0}\right) \in \mathbb{R}^{n_{\mathbf{x}} \times 1}$.

Equation (\ref{F4}) gives the fixed-point representation of (\ref{C1}) as
\begin{equation}\label{F5}
\widetilde{\mathbf{x}}=\mathbf{F}_{\widetilde{\mathbf{u}}_{\mathrm{PCC}}} \widetilde{\mathbf{u}}_{\mathrm{PCC}}+\mathbf{F}_{\tilde{\mathbf{u}}} \widetilde{\mathbf{u}}+\mathbf{F}_{\mathbf{x}} \mathbf{R},
\end{equation}
where $\mathbf{J}=\left.\nabla_{\mathbf{x}} \mathbf{f}(\mathbf{x})\right|_{\mathbf{x}=\mathbf{x}_{0}}\left(\left.\nabla_{\mathbf{x}} \mathbf{f}(\mathbf{x})\right|_{\mathbf{x}=\mathbf{x}_{0}}\right)^{T}$;~$\mathbf{F}_{\widetilde{\mathrm{u}}_\mathrm{P C C}}=\mathbf{J}^{-1}\left(\left.\nabla_{\mathbf{x}} \mathbf{f}(\mathbf{x})\right|_{\mathbf{x}=\mathbf{x}_{0}}\right)^{T} \mathbf{K}_{1} \in \mathbb{R}^{n_{\mathrm{x}} \times n_{\mathrm{PCC}}}$;~$\mathbf{F}_{\tilde{\mathbf{u}}}=\mathbf{J}^{-1}\left(\left.\nabla_{\mathbf{x}} \mathbf{f}(\mathbf{x})\right|_{\mathbf{x}=\mathbf{x}_{0}}\right)^{T} \mathbf{K}_{2} \in \mathbb{R}^{n_{\mathbf{x}} \times n_{\mathbf{u}}}$;~$\mathbf{F}_{\mathbf{x}}=-\mathbf{J}^{-1}\left(\left.\nabla_{\mathbf{x}} \mathbf{f}(\mathbf{x})\right|_{\mathbf{x}=\mathbf{x}_{0}}\right)^{T} \in \mathbb{R}^{n_{\mathbf{x}} \times n_{\mathrm{e}}}$.

Note that the second-order term $\mathbf{R}$ in (\ref{F5}) requires the second derivatives of the elements of $\mathbf{f}(\mathbf{x})$ in (\ref{C1}) which is a compact formulation of (\ref{D1})-(\ref{D4}). Indeed, the second-order terms associated with (\ref{D1})-(\ref{D3}) are zero. This fact allows a reformulation of $\mathbf{R}$ as

\begin{equation}\label{F6}
\mathbf{R}=\left[\begin{array}{ll}
\mathbf{0}^{\left(2 n_{\mathrm{N}}+n_{\mathrm{L}}\right) \times 1} & \widetilde{\mathbf{R}}
\end{array}\right],
\end{equation}
where $\widetilde{\mathbf{R}} \in \mathbb{R}^{n_{\mathrm{L}} \times 1}$ stacks the second-order terms associated with (\ref{D4}): 
\begin{equation}\label{F7}
\widetilde{\mathbf{R}}=-\widetilde{\mathbf{P}}_{\mathrm{L}} * \widetilde{\mathbf{P}}_{\mathrm{L}}-\widetilde{\mathbf{Q}}_{\mathrm{L}} * \widetilde{\mathbf{Q}}_{\mathrm{L}}+\left(\mathbf{e}_{\mathrm{VI}} \widetilde{\mathbf{V}}\right) * \widetilde{\mathbf{I}},
\end{equation}
where $\mathbf{e}_{\mathrm{VI}} \in \mathbb{R}^{n_{\mathrm{L}} \times n_{\mathrm{N}}}$ is the connection matrix between branches and nodes; the operator “*” indicates the Hadamard product, i.e., the component-wise multiplication. 

Finally, the fixed-point representation in (\ref{F5}) can be equivalently compacted as
\begin{equation}\label{F8}
\widetilde{\mathbf{x}}=\mathbf{F}_{\widetilde{\mathbf{u}}_{\mathrm{PCC}}} \widetilde{\mathbf{u}}_{\mathrm{PCC}}+\mathbf{F}_{\widetilde{\mathbf{u}}} \widetilde{\mathbf{u}}+\mathbf{F}_{\widetilde{\mathbf{x}}} \widetilde{\mathbf{R}}.
\end{equation}

\subsubsection{Analysis by Brouwer’s fixed point theorem}

We introduce the Brouwer’s fixed point theorem \cite{30} in coming Theorem 1 to analyze the solvability of the fixed-point representation (\ref{F8}) subject to the safety limits (\ref{C2})-(\ref{C3}).
\begin{theorem} 
Suppose there are 1) a compact and convex set $\widetilde{\mathbf{X}}$ in real space, and 2) a continuous mapping $\mathbf{F}(\widetilde{\mathbf{x}})$ whose dimension is the same as that of $\widetilde{\mathbf{X}}$. If  $\mathbf{F}(\widetilde{\mathbf{x}}) \in \widetilde{\mathbf{X}}$ for all $\widetilde{\mathbf{x}} \in \widetilde{\mathbf{X}}$,  a solution exists in $\widetilde{\mathbf{X}}$ for $\widetilde{\mathbf{x}}=\mathbf{F}(\widetilde{\mathbf{x}})$.
\end{theorem}

For the fixed-point representation (\ref{F8}), the Brouwer’s fixed point theorem can be applied to guarantee its solvability by 1) regarding $\left(\widetilde{\mathbf{u}}_{\mathrm{PCC}}, \widetilde{\mathbf{u}}\right)$ as parameters, 2) $\widetilde{\mathbf{x}}=\left(\widetilde{\mathbf{P}}_{\mathrm{L}}, \widetilde{\mathbf{Q}}_{\mathrm{L}}, \widetilde{\mathbf{V}}, \widetilde{\mathbf{I}}\right)$ as variables, and 3) $\mathbf{F}(\widetilde{\mathbf{x}})$ as $\mathbf{F}_{\widetilde{\mathbf{u}}_{\mathrm{PCC}}} \widetilde{\mathbf{u}}_{\mathrm{PCC}}+\mathbf{F}_{\widetilde{\mathbf{u}}} \widetilde{\mathbf{u}}+\mathbf{F}_{\widetilde{\mathbf{x}}} \widetilde{\mathbf{R}}$. Particularly, the application condition of the Brouwer’s fixed point theorem relies on 1) the selection of a compact and convex set $\widetilde{\mathbf{X}}$, and 2) how to guarantee $\mathbf{F}(\widetilde{\mathbf{x}}) \in \widetilde{\mathbf{X}}$ for all $\widetilde{\mathbf{x}} \in \widetilde{\mathbf{X}}$. We next discuss these two issues in more detail.

\textit{Issue 1:} how to select $\widetilde{\mathbf{X}}$. Generally, any compact and convex set in the domain of $\widetilde{\mathbf{x}}$ can be used in the Brouwer’s fixed point theorem. Particularly, if we select $\widetilde{\mathbf{X}}$ as a subset of the constraints (\ref{C2})-(\ref{C3}),the solution from the Brouwer’s fixed point theorem will inherently satisfy the constraints (\ref{C2})-(\ref{C3}). In this paper, we construct such a set $\widetilde{\mathbf{X}}$ as 
\begin{small}
\begin{equation}\label{X1}
\widetilde{\mathbf{X}} \triangleq\left\{\begin{array}{l|l}
\left(\tilde{\mathbf{P}}_{\mathrm{L}}, \tilde{\mathbf{Q}}_{\mathrm{L}}, \tilde{\mathbf{V}}, \tilde{\mathbf{I}}\right) \in \mathbb{R}^{n_{\mathrm{x}} \times 1} & \begin{array}{l}
\widetilde{\mathbf{V}}_{\min} \leq \widetilde{\mathbf{V}} \leq \widetilde{\mathbf{V}}_{\max} \\
\widetilde{\mathbf{I}}_{\min} \leq \widetilde{\mathbf{I}} \leq \widetilde{\mathbf{I}}_{\max}
\end{array}
\end{array}\right\}
\end{equation}
\end{small}
where $\left(\widetilde{\mathbf{V}}_{\min }, \widetilde{\mathbf{V}}_{\max }, \widetilde{\mathbf{I}}_{\min }, \widetilde{\mathbf{I}}_{\max }\right)$ is restricted by 
\begin{equation}\label{X2}
\left\{\begin{array}{l}
\mathbf{V}_{\min }-\mathbf{V}_{0} \leq \tilde{\mathbf{V}}_{\min } \leq \tilde{\mathbf{V}}_{\max } \leq \mathbf{V}_{\max }-\mathbf{V}_{0} \\
\mathbf{I}_{\min }-\mathbf{I}_{0} \leq \tilde{\mathbf{I}}_{\min } \leq \tilde{\mathbf{I}}_{\max } \leq \mathbf{I}_{\max }-\mathbf{I}_{0}
\end{array}\right..
\end{equation}

Note that $\mathbf{V}_{\min }-\mathbf{V}_{0}$ and $\mathbf{I}_{\min }-\mathbf{I}_{0}$ are non-positive, while $\mathbf{V}_{\max }-\mathbf{V}_{0}$ and $\mathbf{I}_{\max }-\mathbf{I}_{0}$ are non-negative. How to determine $\left(\widetilde{\mathbf{V}}_{\min }, \widetilde{\mathbf{V}}_{\max }, \widetilde{\mathbf{I}}_{\min }, \widetilde{\mathbf{I}}_{\max }\right)$ will be discussed later, since the determination of $\left(\widetilde{\mathbf{V}}_{\min }, \widetilde{\mathbf{V}}_{\max }, \widetilde{\mathbf{I}}_{\min }, \widetilde{\mathbf{I}}_{\max }\right)$ impacts how to guarantee $\mathbf{F}(\widetilde{\mathbf{x}}) \in \widetilde{\mathbf{X}}$ for all $\widetilde{\mathbf{x}} \in \widetilde{\mathbf{X}}$.

\textit{Issue 2:} how to guarantee $\mathbf{F}(\widetilde{\mathbf{x}}) \in \widetilde{\mathbf{X}}$ for all $\widetilde{\mathbf{x}} \in \widetilde{\mathbf{X}}$. To fulfill the solvability of the fixed-point representation (\ref{F8}) in the Brouwer’s fixed point theorem, we develop an approach to show how to guarantee $\mathbf{F}(\widetilde{\mathbf{x}}) \in \widetilde{\mathbf{X}}$ for all $\widetilde{\mathbf{x}} \in \widetilde{\mathbf{X}}$ when the appropriate parameters $\left(\widetilde{\mathbf{u}}_{\mathrm{PCC}}, \widetilde{\mathbf{u}}\right)$ are given. Suppose one can find bounds $\widetilde{\mathbf{R}}_{\text {min}}$ and $\widetilde{\mathbf{R}}_{\text {max}}$ such that $\widetilde{\mathbf{R}}_{\min} \leq \widetilde{\mathbf{R}} \leq \widetilde{\mathbf{R}}_{\max}$ for all $\widetilde{\mathbf{x}} \in \widetilde{\mathbf{X}}$. The method to find $\widetilde{\mathbf{R}}_{\text {min}}$ and $\widetilde{\mathbf{R}}_{\text {max}}$ will be elaborated later. Let $\mathbf{A}_{1} \in \mathbb{R}^{\left(n_{\mathrm{L}}+n_{\mathrm{N}}\right) \times n_{\mathrm{x}}}$ be a selection matrix to make $\mathbf{A}_{1}\left[\begin{array}{llll}\widetilde{\mathbf{P}}_{L}^{T} & \widetilde{\mathbf{Q}}_{\mathrm{L}}^{T} & \widetilde{\mathbf{V}}^{T} & \widetilde{\mathbf{I}}^{T}\end{array}\right]^{T}=\left[\begin{array}{ll}\widetilde{\mathbf{V}}^{T} & \widetilde{\mathbf{I}}^{T}\end{array}\right]^{T}$. Define $\mathbf{L}_{\widetilde{\mathbf{x}}}=\mathbf{A}_{1} \mathbf{F}_{\widetilde{\mathbf{x}}} \in \mathbb{R}^{\left(n_{\mathrm{L}}+n_{\mathrm{N}}\right) \times n_{\mathrm{x}}}$. Define $\mathbf{M}_{\widetilde{\mathbf{x}}}^{+}=\mathbf{L}_{\widetilde{\mathbf{x}}}^{+}-\mathbf{L}_{\widetilde{\mathbf{x}}}^{-} \in\left\{0 \cup \mathbb{R}_{+}^{\left(n_{\mathrm{L}}+n_{\mathrm{N}}\right) \times n_{\mathrm{L}}}\right\}$, where $\mathbf{L}_{\widetilde{\mathbf{x}}}^{+}$ and $\mathbf{L}_{\widetilde{\mathbf{x}}}^{-}$ are respectively the non-negative and non-positive parts of $\mathbf{L}_{\widetilde{\mathbf{x}}}$, with $\mathbf{L}_{\widetilde{\mathbf{x}}}=\mathbf{L}_{\widetilde{\mathbf{x}}}^{+}+\mathbf{L}_{\widetilde{\mathbf{x}}}^{-}$.

\begin{prop} 
If the following inequality holds:
\begin{equation}\label{P1-1}
\begin{aligned}
&\left[\begin{array}{cc}
\widetilde{\mathbf{V}}_{\min }^{T} & \widetilde{\mathbf{I}}_{\min }^{T}
\end{array}\right]^{T}-\left[\begin{array}{cc}
\widetilde{\mathbf{V}}_{\max }^{T} & \widetilde{\mathbf{I}}_{\max }^{T}
\end{array}\right]^{T}\\
&\leq \mathbf{M}_{\widetilde{\mathbf{x}}}^{+} \widetilde{\mathbf{R}}_{\min }-\mathbf{M}_{\widetilde{\mathbf{x}}}^{+} \widetilde{\mathbf{R}}_{\max }
\end{aligned}.
\end{equation}
then there exists a vector $\mathbf{U} \in \mathbb{R}^{n_{\mathrm{x}} \times 1}$ such that $\mathbf{U}+\mathbf{F}_{\widetilde{\mathbf{x}}} \widetilde{\mathbf{R}}$ in $\widetilde{\mathbf{X}}$ for all $\widetilde{\mathbf{x}} \in \widetilde{\mathbf{X}}$. If we make $\mathbf{U}=\mathbf{F}_{\widetilde{\mathbf{u}}_\mathrm{PCC}} \widetilde{\mathbf{u}}_{\mathrm{PCC}}+\mathbf{F}_{\widetilde{\mathbf{u}}} \widetilde{\mathbf{u}}$, then $\mathbf{F}(\widetilde{\mathbf{x}}) \in \widetilde{\mathbf{X}}$ for all $\widetilde{\mathbf{x}} \in \widetilde{\mathbf{X}}$ can be fulfilled.
\end{prop}

\begin{proof}
The fulfillment of the inequality  (\ref{P1-1}) yields
\begin{equation}\label{P1-2}
\begin{aligned}
&{\left[\begin{array}{ll}
\widetilde{\mathbf{V}}_{\min }^{T} & \widetilde{\mathbf{I}}_{\min }^{T}
\end{array}\right]^{T}-\mathbf{L}_{\widetilde{\mathbf{x}}}^{+} \widetilde{\mathbf{R}}_{\min }-\mathbf{L}_{\widetilde{\mathbf{x}}}^{-} \widetilde{\mathbf{R}}_{\max }} \\
&\leq\left[\begin{array}{cc}
\widetilde{\mathbf{V}}_{\max }^{T} & \widetilde{\mathbf{I}}_{\max }^{T}
\end{array}\right]^{T}-\mathbf{L}_{\widetilde{\mathbf{x}}}^{-} \widetilde{\mathbf{R}}_{\min }-\mathbf{L}_{\widetilde{\mathbf{x}}}^{+} \widetilde{\mathbf{R}}_{\max }
\end{aligned}.
\end{equation}

Then, there exists a vector $\mathbf{U}$ such that
\begin{equation}\label{P1-3}
\begin{aligned}
&\begin{cases}{\left[\begin{array}{ll}
\widetilde{\mathbf{V}}_{\min }^{T} & \widetilde{\mathbf{I}}_{\min }^{T}
\end{array}\right]^{T}-\mathbf{L}_{\widetilde{\mathbf{x}}}^{-} \widetilde{\mathbf{R}}_{\max }-\mathbf{L}_{\widetilde{\mathbf{x}}}^{+} \widetilde{\mathbf{R}}_{\min } \leq \mathbf{A}_{1} \mathbf{U}} \\
\mathbf{A}_{1} \mathbf{U} \leq\left[\begin{array}{cc}
\widetilde{\mathbf{V}}_{\max }^{T} & \widetilde{\mathbf{I}}_{\max }^{T}
\end{array}\right]^{T}-\mathbf{L}_{\widetilde{\mathbf{x}}}^{+} \widetilde{\mathbf{R}}_{\max }-\mathbf{L}_{\widetilde{\mathbf{x}}}^{-} \widetilde{\mathbf{R}}_{\min }\end{cases}\\
&\Rightarrow\left\{\begin{array}{l}
{\left[\begin{array}{lc}
\widetilde{\mathbf{V}}_{\min }^{T} & \widetilde{\mathbf{I}}_{\min }^{T}
\end{array}\right]^{T} \leq \mathbf{A}_{1} \mathbf{U}+\mathbf{L}_{\widetilde{\mathbf{x}}}^{-} \widetilde{\mathbf{R}}_{\max }+\mathbf{L}_{\widetilde{\mathbf{x}}}^{+} \widetilde{\mathbf{R}}_{\min }} \\
\mathbf{A}_{1} \mathbf{U}+\mathbf{L}_{\widetilde{\mathbf{x}}}^{+} \widetilde{\mathbf{R}}_{\max }+\mathbf{L}_{\widetilde{\mathbf{x}}}^{-} \widetilde{\mathbf{R}}_{\min } \leq\left[\begin{array}{cc}
\widetilde{\mathbf{V}}_{\max }^{T} & \widetilde{\mathbf{I}}_{\max }^{T}
\end{array}\right]^{T}
\end{array}\right.
\end{aligned}.
\end{equation}

Since $\mathbf{L}_{\widetilde{\mathbf{x}}}^{+}$ is non-negative and $\mathbf{L}_{\widetilde{\mathbf{x}}}^{-}$ is non-positive, $\mathbf{L}_{\widetilde{\mathbf{x}}} \widetilde{\mathbf{R}}=\left(\mathbf{L}_{\widetilde{\mathbf{x}}}^{+}+\mathbf{L}_{\widetilde{\mathbf{x}}}^{-}\right) \widetilde{\mathbf{R}}$ can be bounded by
\begin{equation}\label{P1-4}
\mathbf{L}_{\widetilde{\mathbf{x}}}^{-} \widetilde{\mathbf{R}}_{\max }+\mathbf{L}_{\widetilde{\mathbf{x}}}^{+} \widetilde{\mathbf{R}}_{\min } \leq \mathbf{L}_{\widetilde{\mathbf{x}}} \widetilde{\mathbf{R}} \leq \mathbf{L}_{\widetilde{\mathbf{x}}}^{+} \widetilde{\mathbf{R}}_{\max }+\mathbf{L}_{\widetilde{\mathbf{x}}}^{-} \widetilde{\mathbf{R}}_{\min }.
\end{equation}

Combining (\ref{P1-3}) and (\ref{P1-4}) yields $\mathbf{U}+\mathbf{F}_{\widetilde{\mathbf{x}}} \widetilde{\mathbf{R}}$ in $\widetilde{\mathbf{X}}$ for all $\widetilde{\mathbf{x}} \in \widetilde{\mathbf{X}}$, i.e., Proposition 1 holds.
\end{proof}

\subsubsection{Linear constraints to characterize an AC-feasible sub-region}

Proposition 1 inspires the construction of the linear constraints to characterize a sub-region of the AC-feasible power transfer region. As shown in (\ref{P1-3}), once $\left(\widetilde{\mathbf{V}}_{\min }, \widetilde{\mathbf{V}}_{\max }, \widetilde{\mathbf{I}}_{\min }, \widetilde{\mathbf{I}}_{\max }\right)$ and $\left(\widetilde{\mathbf{R}}_{\min }, \widetilde{\mathbf{R}}_{\max }\right)$ that satisfy (\ref{P1-1}) are given, the following polytope $\mathbf{\widetilde{\Omega}}$ composed of linear constraints in the domain of $\left(\widetilde{\mathbf{u}}_{\mathrm{PCC}}, \widetilde{\mathbf{u}}\right)$ can be given 
\begin{footnotesize}
\begin{equation}\label{Om1}
\widetilde{\mathbf{\Omega}} \triangleq\left\{\begin{array}{c|l}
\left(\widetilde{\mathbf{u}}_{\mathrm{PCC}}, \widetilde{\mathbf{u}}\right) & \mathbf{b}_{\mathrm{min}} \leq \mathbf{A}_{1} \mathbf{F}_{\widetilde{\mathbf{u}} \mathrm{PCC}} \widetilde{\mathbf{u}}_{\mathrm{PCC}}+\mathbf{A}_{1} \mathbf{F}_{\widetilde{\mathbf{u}}} \widetilde{\mathbf{u}} \\
\in \mathbb{R}^{\left(n_{\mathrm{PCC}}+n_{\mathrm{u}}\right) \times 1} & \mathbf{A}_{1} \mathbf{F}_{\widetilde{\mathbf{u}}_{\mathrm{PCC}}} \widetilde{\mathbf{u}}_{\mathrm{PCC}}+\mathbf{A}_{1} \mathbf{F}_{\widetilde{\mathbf{u}}} \widetilde{\mathbf{u}} \leq \mathbf{b}_{\mathrm{max}}
\end{array}\right\},
\end{equation}
\end{footnotesize}
\noindent where $\mathbf{b}_{\min }=\left[\begin{array}{ll}\widetilde{\mathbf{V}}_{\min }^{T} & \widetilde{\mathbf{I}}_{\min }^{T}\end{array}\right]^{T}-\mathbf{L}_{\widetilde{\mathbf{x}}}^{-} \widetilde{\mathbf{R}}_{\max }-\mathbf{L}_{\widetilde{\mathbf{x}}}^{+} \widetilde{\mathbf{R}}_{\min }$ and $\mathbf{b}_{\max }=\left[\begin{array}{ll}\widetilde{\mathbf{V}}_{\max }^{T} & \widetilde{\mathbf{I}}_{\max }^{T}\end{array}\right]^{T}-\mathbf{L}_{\widetilde{\mathbf{x}}}^{+} \widetilde{\mathbf{R}}_{\max }-\mathbf{L}_{\widetilde{\mathbf{x}}}^{-} \widetilde{\mathbf{R}}_{\min }$.

If we give $\left(\widetilde{\mathbf{u}}_{\mathrm{PCC}}, \widetilde{\mathbf{u}}\right)$ in $\mathbf{\widetilde{\Omega}}$, the Brouwers’s fixed point theorem can be used to guarantee the solvability of the fixed-point representation (\ref{F8}), i.e., the power flow equation (\ref{C1}) can be satisfied. Furthermore, since $\widetilde{\mathbf{X}}$ in the Brouwers’s fixed point theorem is a subset of the safety limits (\ref{C2})-(\ref{C3}), the solution in (\ref{C1}) from the  Brouwer’s fixed point theorem also satisfies (\ref{C2})-(\ref{C3}). If we further require an additional bound $\mathbf{u}_{\min }-\mathbf{u}_{0} \leqslant \widetilde{\mathbf{u}} \leqslant \mathbf{u}_{\max }-\mathbf{u}_{0}$, the constraint (\ref{C4}) is also satisfied. Consequently, $\mathbf{\widetilde{\Omega}}$ intersecting with $\mathbf{u}_{\min }-\mathbf{u}_{0} \leqslant \widetilde{\mathbf{u}} \leqslant \mathbf{u}_{\max }-\mathbf{u}_{0}$ guarantees AC-feasibility in the domain of $\left(\widetilde{\mathbf{u}}_{\mathrm{PCC}}, \widetilde{\mathbf{u}}\right)$. If we project it onto the space of $\widetilde{\mathbf{u}}_{\mathrm{PCC}}$, the resulting projection is obviously a sub-region of the AC-feasible power transfer region in the domain of $\widetilde{\mathbf{u}}_{\mathrm{PCC}}$. Such a projection can be easily implemented by the current methods \cite{6}-\cite{12} which work for linear models. The obtained sub-region in the domain of $\widetilde{\mathbf{u}}_{\mathrm{PCC}}$ can be easily transformed to $\mathbf{u}_{\mathrm{PCC}}$ by $\widetilde{\mathbf{u}}_{\mathrm{PCC}}=\left(\mathbf{u}_\mathrm{P C C}-\mathbf{u}_{\mathrm{PCC} 0}\right)$, , since all the constraints in the sub-region are linear.

Consequently, the key to construct our linear constraints lies in calculating $\left(\widetilde{\mathbf{V}}_{\min }, \widetilde{\mathbf{V}}_{\max }, \widetilde{\mathbf{I}}_{\min }, \widetilde{\mathbf{I}}_{\max }\right)$ and $\left(\widetilde{\mathbf{R}}_{\min }, \widetilde{\mathbf{R}}_{\max }\right)$ that satisfy the constraint (\ref{P1-1}). Recall that $\widetilde{\mathbf{R}}_{\text {min }}$ and $\widetilde{\mathbf{R}}_{\text {max }}$ are lower and upper bounds of the second-order term $\widetilde{\mathbf{R}}$ (\ref{F7}) when $\left(\mathbf{P}_{\mathrm{L}}, \mathbf{Q}_{\mathrm{L}}, \mathbf{V}, \mathbf{I}\right)$ varies over $\widetilde{\mathbf{X}}$ defined in (\ref{X1}). $\widetilde{\mathbf{R}}$ is composed of three terms: two quadratic terms $-\widetilde{\mathbf{P}}_{\mathrm{L}} * \widetilde{\mathbf{P}}_{\mathrm{L}}$ and $-\widetilde{\mathbf{Q}}_{\mathrm{L}} * \widetilde{\mathbf{Q}}_{\mathrm{L}}$, and one bilinear term $\mathbf{e}_{\mathrm{V} \mathrm{I}} \tilde{\mathbf{V}} * \tilde{\mathbf{I}}$. Note that $\left(\widetilde{\mathbf{P}}_{\mathrm{L}}, \widetilde{\mathbf{Q}}_{\mathrm{L}}, \widetilde{\mathbf{V}}, \tilde{\mathbf{I}}\right)$ is decoupled in $\widetilde{\mathbf{X}}$ defined in (\ref{X1}). Consequently, $\widetilde{\mathbf{R}}_{\max }$ and $\widetilde{\mathbf{R}}_{\min }$ can be decomposed into three terms as
\begin{equation}\label{Rmax1}
\widetilde{\mathbf{R}}_{\max }=\widetilde{\mathbf{R}}_{\max }^{(\mathrm{P})}+\widetilde{\mathbf{R}}_{\max }^{(\mathrm{Q})}+\widetilde{\mathbf{R}}_{\max }^{(\mathrm{VI})},
\end{equation}
\begin{equation}\label{Rmin1}
\widetilde{\mathbf{R}}_{\min }=\widetilde{\mathbf{R}}_{\min }^{(\mathrm{P})}+\widetilde{\mathbf{R}}_{\min }^{(\mathrm{Q})}+\widetilde{\mathbf{R}}_{\min }^{(\mathrm{VI})},
\end{equation}
where $\left(\widetilde{\mathbf{R}}_{\max }^{(\mathrm{P}}, \widetilde{\mathbf{R}}_{\max }^{(\mathrm{Q})}, \widetilde{\mathbf{R}}_{\max }^{(\mathrm{VI})}, \widetilde{\mathbf{R}}_{\min }^{(\mathrm{P})}, \widetilde{\mathbf{R}}_{\min }^{(\mathrm{Q})}, \widetilde{\mathbf{R}}_{\min }^{(\mathrm{VI})}\right)$ can be formulated based on their monotonicity when $\left(\widetilde{\mathbf{P}}_{\mathrm{L}}, \widetilde{\mathbf{Q}}_{\mathrm{L}}, \widetilde{\mathbf{V}}, \tilde{\mathbf{I}}\right)$ varies over $\widetilde{\mathbf{X}}$, as described below.
\begin{equation}\label{RmaxP}
\left\{\begin{array}{l}
\widetilde{\mathbf{R}}_{\max }^{(\mathrm{P})}=\mathbf{0}, \text { if } \widetilde{\mathbf{P}}_{\mathrm{Lmin}} \leq \mathbf{0} \leq \widetilde{\mathbf{P}}_{\mathrm{L} \text { max }} \\
\widetilde{\mathbf{R}}_{\text {max }}^{(\mathrm{P})}=-\widetilde{\mathbf{P}}_{\mathrm{L} \text { min }} * \widetilde{\mathbf{P}}_{\mathrm{L} \text { min }}, \text { if } \mathbf{0} \leq \widetilde{\mathbf{P}}_{\mathrm{Lmin}} \leq \widetilde{\mathbf{P}}_{\mathrm{L} \text { max }} \\
\widetilde{\mathbf{R}}_{\text {max }}^{(\mathrm{P})}=-\widetilde{\mathbf{P}}_{\mathrm{L} \text { max }} * \widetilde{\mathbf{P}}_{\mathrm{L} \text { max }}, \text { if } \widetilde{\mathbf{P}}_{\mathrm{Lmin}} \leq \widetilde{\mathbf{P}}_{\mathrm{L} \text { max }} \leq \mathbf{0}
\end{array}\right.,
\end{equation}
\begin{equation}\label{RmaxQ}
\left\{\begin{array}{l}
\widetilde{\mathbf{R}}_{\max }^{(\mathrm{Q})}=\mathbf{0}, \text { if } \widetilde{\mathbf{Q}}_{\mathrm{Lmin}} \leq \mathbf{0} \leq \widetilde{\mathbf{Q}}_{\mathrm{L} \max } \\
\widetilde{\mathbf{R}}_{\max }^{(\mathrm{Q})}=-\widetilde{\mathbf{Q}}_{\mathrm{L} \min } * \widetilde{\mathbf{Q}}_{\mathrm{L} \text { min }}, \text { if } \mathbf{0} \leq \widetilde{\mathbf{Q}}_{\mathrm{Lmin}} \leq \widetilde{\mathbf{Q}}_{\mathrm{L} \max } \\
\widetilde{\mathbf{R}}_{\max }^{(\mathrm{Q})}=-\widetilde{\mathbf{Q}}_{\mathrm{L} \max } * \widetilde{\mathbf{Q}}_{\mathrm{L} \max }, \text { if } \widetilde{\mathbf{Q}}_{\mathrm{Lmin}} \leq \widetilde{\mathbf{Q}}_{\mathrm{L} \max } \leq \mathbf{0}
\end{array}\right.,
\end{equation}
\begin{equation}\label{RminP}
\left\{\begin{array}{l}
\widetilde{\mathbf{R}}_{\min }^{(\mathrm{P})}=\left\{-\widetilde{\mathbf{P}}_{\mathrm{Lmin}} * \widetilde{\mathbf{P}}_{\mathrm{Lmin}},-\widetilde{\mathbf{P}}_{\mathrm{L} \max } * \widetilde{\mathbf{P}}_{\mathrm{L} \max }\right\}^{\min }, \\
\text { if } \widetilde{\mathbf{P}}_{\mathrm{Lmin}} \leq \mathbf{0} \leq \widetilde{\mathbf{P}}_{\mathrm{L} \max } \\
\widetilde{\mathbf{R}}_{\min} ^{(\mathrm{P})}=-\widetilde{\mathbf{P}}_{\mathrm{L} \max } * \widetilde{\mathbf{P}}_{\mathrm{L} \max }, \text { if } \mathbf{0} \leq \widetilde{\mathbf{P}}_{\mathrm{Lmin}} \leq \widetilde{\mathbf{P}}_{\mathrm{Lmax}} \\
\widetilde{\mathbf{R}}_{\mathrm{min}}^{(\mathrm{P})}=-\widetilde{\mathbf{P}}_{\mathrm{Lmin}} * \widetilde{\mathbf{P}}_{\mathrm{Lmin}}, \text { if } \widetilde{\mathbf{P}}_{\mathrm{Lmin}} \leq \widetilde{\mathbf{P}}_{\mathrm{L} \max } \leq \mathbf{0}
\end{array}\right.,
\end{equation}
\begin{equation}\label{RminQ}
\left\{\begin{array}{l}
\widetilde{\mathbf{R}}_{\text {min }}^{(\mathrm{Q})}=\left\{-\widetilde{\mathbf{Q}}_{\mathrm{Lmin}} * \widetilde{\mathbf{Q}}_{\mathrm{Lmin}},-\widetilde{\mathbf{Q}}_{\mathrm{Lmax}} * \widetilde{\mathbf{Q}}_{\mathrm{Lmax}}\right\}^{\min}, \\
\text { if } \widetilde{\mathbf{Q}}_{\mathrm{Lmin}} \leq \mathbf{0} \leq \widetilde{\mathbf{Q}}_{\mathrm{Lmax}} \\
\widetilde{\mathbf{R}}_{\min }^{(\mathrm{Q})}=-\widetilde{\mathbf{Q}}_{\mathrm{Lmax}} * \widetilde{\mathbf{Q}}_{\mathrm{Lmax}}, \text { if } \mathbf{0} \leq \widetilde{\mathbf{Q}}_{\mathrm{Lmin}} \leq \widetilde{\mathbf{Q}}_{\mathrm{Lmax}} \\
\widetilde{\mathbf{R}}_{\min }^{(\mathrm{Q})}=-\widetilde{\mathbf{Q}}_{\mathrm{Lmin}} * \widetilde{\mathbf{Q}}_{\mathrm{Lmin}}, \text { if } \widetilde{\mathbf{Q}}_{\mathrm{Lmin}} \leq \widetilde{\mathbf{Q}}_{\mathrm{Lmax}} \leq \mathbf{0}
\end{array}\right.,
\end{equation}
\begin{equation}\label{RmaxVI}
\widetilde{\mathbf{R}}_{\max }^{(\mathrm{VI})}=\left\{\begin{array}{l}
\left(\mathbf{e}_{\mathrm{VI}} \widetilde{\mathbf{V}}_{\min }\right) * \widetilde{\mathbf{I}}_{\min },\left(\mathbf{e}_{\mathrm{VI}} \widetilde{\mathbf{V}}_{\min }\right) * \widetilde{\mathbf{I}}_{\max } \\
\left(\mathbf{e}_{\mathrm{VI}} \widetilde{\mathbf{V}}_{\max }\right) * \widetilde{\mathbf{I}}_{\min },\left(\mathbf{e}_{\mathrm{VI}} \widetilde{\mathbf{V}}_{\max }\right) * \widetilde{\mathbf{I}}_{\max }
\end{array}\right\}^{\max },
\end{equation}
\begin{equation}\label{RminVI}
\widetilde{\mathbf{R}}_{\min }^{(\mathrm{VI})}=\left\{\begin{array}{l}
\left(\mathbf{e}_{\mathrm{VI}} \widetilde{\mathbf{V}}_{\min }\right) * \widetilde{\mathbf{I}}_{\min },\left(\mathbf{e}_{\mathrm{VI}} \widetilde{\mathbf{V}}_{\min }\right) * \widetilde{\mathbf{I}}_{\max } \\
\left(\mathbf{e}_{\mathrm{VI}} \widetilde{\mathbf{V}}_{\max }\right) * \widetilde{\mathbf{I}}_{\min },\left(\mathbf{e}_{\mathrm{VI}} \widetilde{\mathbf{V}}_{\max }\right) * \widetilde{\mathbf{I}}_{\max }
\end{array}\right\}^{\min },
\end{equation}

\noindent where $\widetilde{\mathbf{P}}_{\text {Lmin }}$ and $\widetilde{\mathbf{P}}_{\text {Lmax }}$ are lower and upper bounds of $\widetilde{\mathbf{P}}_{\mathrm{L}}$; $\widetilde{\mathbf{Q}}_{\mathrm{Lmin}}$ and $\widetilde{\mathbf{Q}}_{\mathrm{Lmax}}$ are lower and upper bounds of $\widetilde{\mathbf{Q}}_{\mathrm{L}}$; the operator  $``{\max }"$ (resp. $``{\min }"$) selects the maximum (resp. minimum) element.

Substituting (\ref{Rmax1})-(\ref{RminVI}) into (\ref{Om1}) returns a region $\widetilde{\mathbf{\Omega}}$ of $\left(\widetilde{\mathbf{u}}_{\mathrm{PCC}}, \widetilde{\mathbf{u}}\right)$, in which the condition (\ref{P1-1}) for AC-feasibility can be satisfied. In this paper, we try to expand $\widetilde{\mathbf{\Omega}}$ as large as possible. For that purpose, we firstly quantify how large the polytope $\widetilde{\mathbf{\Omega}}$ is. Based on (\ref{P1-1}), $\widetilde{\mathbf{\Omega}}$ is shaped by parallel hyperplanes in the domain of $\left(\widetilde{\mathbf{u}}_{\mathrm{PCC}}, \widetilde{\mathbf{u}}\right)$. Consequently, the sum of distances between these hyperplanes across all the dimensions is taken as an index to reflect how large the polytope $\widetilde{\mathbf{\Omega}}$ is. Such a sum can be calculated as (\ref{OBJ1}) shown at the bottom of this page.
\begin{figure*}[hb]
\hrulefill
\begin{equation}\label{OBJ1}
\begin{aligned}
&-\mathbf{e}_{1}^{T}\left\{\left\{\left(\left[\begin{array}{cc}
\widetilde{\mathbf{V}}_{\min }^{T} & \widetilde{\mathbf{I}}_{\min }^{T}
\end{array}\right]^{T}-\mathbf{L}_{\widetilde{\mathbf{x}}}^{+} \widetilde{\mathbf{R}}_{\min }-\mathbf{L}_{\widetilde{\mathbf{x}}}^{-} \widetilde{\mathbf{R}}_{\max }\right)-\left(\left[\begin{array}{cc}
\widetilde{\mathbf{V}}_{\max }^{T} & \widetilde{\mathbf{I}}_{\max }^{T}
\end{array}\right]^{T}-\mathbf{L}_{\widetilde{\mathbf{x}}}^{-} \widetilde{\mathbf{R}}_{\min }-\mathbf{L}_{\widetilde{\mathbf{x}}}^{+} \widetilde{\mathbf{R}}_{\max }\right)\right\} \cdot / \mathbf{H}\right\} \\
&\Rightarrow-\mathbf{e}_{1}^{T}\left\{\left\{\left[\begin{array}{cc}
\widetilde{\mathbf{V}}_{\min }^{T} & \widetilde{\mathbf{I}}_{\min }^{T}
\end{array}\right]^{T}-\left[\begin{array}{cc}
\widetilde{\mathbf{V}}_{\max }^{T} & \widetilde{\mathbf{I}}_{\max }^{T}
\end{array}\right]^{T}+\mathbf{M}_{\widetilde{\mathbf{x}}}^{+} \widetilde{\mathbf{R}}_{\max }-\mathbf{M}_{\widetilde{\mathbf{x}}}^{+} \widetilde{\mathbf{R}}_{\min }\right\} \cdot / \mathbf{H}\right\}
\end{aligned},
\end{equation}
where $\mathbf{e}_{1} \in \mathbf{1}^{\left(n_{\mathrm{L}}+n_{\mathrm{N}}\right) \times 1}$; the operator “./” indicates the component-wise division; $\mathbf{H} \in \mathbb{R}_{+}^{\left(n_{L}+n_{N}\right) \times 1}$ whose $i^{t h}$ element is the root-sum-squares of the $i^{t h}$ row of $\mathbf{L}_{\mathrm{u}}=\left[\begin{array}{ll}\mathbf{A}_{1} \mathbf{F}_{\widetilde{\mathbf{u}}_\mathrm{P C C}} & \mathbf{A}_{1} \mathbf{F}_{\widetilde{\mathbf{u}}}\end{array}\right]$.
\end{figure*}

Based on the rationale above, we formulate the following optimization problem to expand $\widetilde{\mathbf{\Omega}}$ as large as possible.

\noindent \textbf{OP:}
\begin{small}
\begin{equation}\label{OP-OBJ}
\min \mathbf{e}_{1}^{T}\left\{\left\{\begin{array}{l}
{\left[\begin{array}{ll}
\widetilde{\mathbf{V}}_{\min }^{T} & \widetilde{\mathbf{I}}_{\min }^{T}
\end{array}\right]^{T}-\left[\begin{array}{ll}
\widetilde{\mathbf{V}}_{\max }^{T} & \widetilde{\mathbf{I}}_{\max }^{T}
\end{array}\right]^{T}} \\
+\mathbf{M}_{\widetilde{\mathbf{x}}}^{+}\left(\widetilde{\mathbf{R}}_{\max }^{(\mathrm{P})}+\widetilde{\mathbf{R}}_{\max }^{(\mathrm{Q})}+\widetilde{\mathbf{R}}_{\max }^{(\mathrm{VI})}\right) \\
-\mathbf{M}_{\widetilde{\mathbf{x}}}^{+}\left(\widetilde{\mathbf{R}}_{\min }^{(\mathrm{P})}+\widetilde{\mathbf{R}}_{\min }^{(\mathrm{Q})}+\widetilde{\mathbf{R}}_{\min }^{(\mathrm{VI})}\right)
\end{array}\right\} . / \mathbf{H}\right\},
\end{equation}
\begin{equation}\label{OP-Cons}
\text{s.t. Constraints (\ref{X2})-(\ref{P1-1}) and (\ref{RmaxP})-(\ref{RminVI})} ,
\end{equation}
\end{small}
\noindent over $\widetilde{\mathbf{V}}_{\min }$, $\widetilde{\mathbf{V}}_{\max }$, $\widetilde{\mathbf{I}}_{\min }$, $\widetilde{\mathbf{I}}_{\max }$, $\widetilde{\mathbf{P}}_{\text {Lmin }}$, $\widetilde{\mathbf{P}}_{\text {Lmax }}$, $\widetilde{\mathbf{Q}}_{\mathrm{Lmin}}$, $\widetilde{\mathbf{Q}}_{\mathrm{Lmax}}$, $\widetilde{\mathbf{R}}_{\min }^{(\mathrm{P})}$, $\widetilde{\mathbf{R}}_{\max }^{(\mathrm{P})}$, $\widetilde{\mathbf{R}}_{\min }^{(\mathrm{Q})}$, $\widetilde{\mathbf{R}}_{\max }^{(\mathrm{Q})}$, $\widetilde{\mathbf{R}}_{\min }^{(\mathrm{VI})}$, $\widetilde{\mathbf{R}}_{\max }^{(\mathrm{VI})}$. Note that the objective function is to minimize the negative sum of distances in (\ref{OBJ1}), i.e., to maximize the sum of distances. Also, $\left(\widetilde{\mathbf{R}}_{\min }, \widetilde{\mathbf{R}}_{\max }\right)$ is eliminated by (\ref{Rmax1})-(\ref{Rmin1}).

The operators “\{\}$_{\max }$” and “\{\}$_{\min }$” in the constraint (\ref{OP-Cons}) can be removed by additionally introducing binary variables, as done in certain commercial solvers (e.g., GUROBI). Also, the explicit formulation of $\left(\widetilde{\mathbf{R}}_{\max }^{(\mathrm{P})}, \widetilde{\mathbf{R}}_{\min }^{(\mathrm{P})}\right)$ requires the binary variables to indicate the rank among $\mathbf{0}, \widetilde{\mathbf{P}}_{\mathrm{Lmin}}$, and $\widetilde{\mathbf{P}}_{\mathrm{Lmax}}$. A similar situation occurs for the explicit formulation of $\left(\widetilde{\mathbf{R}}_{\max }^{(\mathrm{Q})}, \widetilde{\mathbf{R}}_{\min }^{(\mathrm{Q})}\right)$. These issues cast the OP problem as a MINLP which is difficult to solve. To avoid binary variables, we develop the following NLP to solve part of the variables in the OP problem (how to obtain the rest of the variables will be explained later in Proposition 2):

\noindent \textbf{EP:}
\begin{small}
\begin{equation}\label{EP-OBJ}
\min \mathbf{e}_{1}^{T}\left\{\left\{\begin{array}{l}
{\left[\begin{array}{ll}
\widetilde{\mathbf{V}}_{\min }^{T} & \widetilde{\mathbf{I}}_{\min }^{T}
\end{array}\right]^{T}-\left[\begin{array}{ll}
\widetilde{\mathbf{V}}_{\max }^{T} & \widetilde{\mathbf{I}}_{\max }^{T}
\end{array}\right]^{T}} \\
+\mathbf{M}_{\widetilde{\mathbf{x}}}^{+}\widetilde{\mathbf{R}}_{\max }^{(\mathrm{VI})}-\mathbf{M}_{\widetilde{\mathbf{x}}}^{+}\widetilde{\mathbf{R}}_{\min }^{(\mathrm{VI})}
\end{array}\right\} . / \mathbf{H}\right\},
\end{equation}
\end{small}
\begin{equation}\label{EP-C1}
\text{s.t.} \left\{\begin{array}{l}
\mathbf{V}_{\min }-\mathbf{V}_{0} \leq \widetilde{\mathbf{V}}_{\min } \leq \widetilde{\mathbf{V}}_{\max } \leq \mathbf{V}_{\max }-\mathbf{V}_{0} \\
\mathbf{I}_{\min }-\mathbf{I}_{0} \leq \widetilde{\mathbf{I}}_{\min } \leq \widetilde{\mathbf{I}}_{\max } \leq \mathbf{I}_{\max }-\mathbf{I}_{0}
\end{array}\right.,
\end{equation}
\begin{small}
\begin{equation}\label{EP-C2}
\left\{\begin{array}{l}
\widetilde{\mathbf{R}}_{\min }^{(\mathrm{VI})} \leq\left(\mathbf{e}_{\mathrm{VI}} \widetilde{\mathbf{V}}_{\min }\right) * \widetilde{\mathbf{I}}_{\min }, \widetilde{\mathbf{R}}_{\min }^{(\mathrm{VI})} \leq\left(\mathbf{e}_{\mathrm{VI}} \widetilde{\mathbf{V}}_{\max }\right) * \widetilde{\mathbf{I}}_{\min } \\
\widetilde{\mathbf{R}}_{\min }^{(\mathrm{VI})} \leq\left(\mathbf{e}_{\mathrm{VI}} \widetilde{\mathbf{V}}_{\min }\right) * \widetilde{\mathbf{I}}_{\max }, \widetilde{\mathbf{R}}_{\min }^{(\mathrm{VI})} \leq\left(\mathbf{e}_{\mathrm{VI}} \widetilde{\mathbf{V}}_{\max }\right) * \widetilde{\mathbf{I}}_{\max }
\end{array}\right.,
\end{equation}
\end{small}
\begin{small}
\begin{equation}\label{EP-C3}
\left\{\begin{array}{l}
\widetilde{\mathbf{R}}_{\max }^{(\mathrm{VI})} \geq\left(\mathbf{e}_{\mathrm{VI}} \widetilde{\mathbf{V}}_{\min }\right) * \widetilde{\mathbf{I}}_{\min }, \widetilde{\mathbf{R}}_{\max }^{(\mathrm{VI})} \geq\left(\mathbf{e}_{\mathrm{VI}} \widetilde{\mathbf{V}}_{\max }\right) * \widetilde{\mathbf{I}}_{\min } \\
\widetilde{\mathbf{R}}_{\max }^{(\mathrm{VI})} \geq\left(\mathbf{e}_{\mathrm{VI}} \widetilde{\mathbf{V}}_{\min }\right) * \widetilde{\mathbf{I}}_{\max }, \widetilde{\mathbf{R}}_{\max }^{(\mathrm{VI})} \geq\left(\mathbf{e}_{\mathrm{VI}} \widetilde{\mathbf{V}}_{\max }\right) * \widetilde{\mathbf{I}}_{\max }
\end{array}\right.,
\end{equation}
\end{small}
\noindent over $\widetilde{\mathbf{V}}_{\min }, \widetilde{\mathbf{V}}_{\max }, \widetilde{\mathbf{I}}_{\min }, \widetilde{\mathbf{I}}_{\max }, \widetilde{\mathbf{R}}_{\min }^{(\mathrm{VI})}, \widetilde{\mathbf{R}}_{\max }^{(\mathrm{VI})}$. Let $\widetilde{\mathbf{x}}_{\mathrm{VI}}^{*}=\left(\widetilde{\mathbf{V}}_{\min }^{*}, \widetilde{\mathbf{V}}_{\max }^{*}, \widetilde{\mathbf{I}}_{\min }^{*}, \widetilde{\mathbf{I}}_{\max }^{*}, \widetilde{\mathbf{R}}_{\min }^{(\mathrm{VI})^{*}}, \widetilde{\mathbf{R}}_{\max }^{(\mathrm{VI}) *}\right)$ denote the optimal solution of the EP problem.

\begin{prop} 
The variables of the OP problem can be divided into two parts:
\begin{equation}\label{Solution}
\left\{\begin{array}{l}
\widetilde{\mathbf{x}}_{\mathrm{VI}}=\left(\widetilde{\mathbf{V}}_{\min }, \widetilde{\mathbf{V}}_{\max }, \widetilde{\mathbf{I}}_{\min }, \widetilde{\mathbf{I}}_{\max }, \widetilde{\mathbf{R}}_{\min }^{(\mathrm{VI})}, \widetilde{\mathbf{R}}_{\max }^{(\mathrm{VI})}\right) \\
\widetilde{\mathbf{x}}_{\mathrm {Line }}=\left(\begin{array}{c}
\widetilde{\mathbf{P}}_{\operatorname{Lmin}}, \widetilde{\mathbf{P}}_{\mathrm{Lmax}}, \widetilde{\mathbf{Q}}_{\mathrm{Lmin}}, \widetilde{\mathbf{Q}}_{\mathrm{Lmax}}, \\
\widetilde{\mathbf{R}}_{\min }^{(\mathrm{P})}, \widetilde{\mathbf{R}}_{\max }^{(\mathrm{P})}, \widetilde{\mathbf{R}}_{\min }^{(\mathrm{Q})}, \widetilde{\mathbf{R}}_{\max }^{(\mathrm{Q})}
\end{array}\right)
\end{array}\right..
\end{equation}

If 1) $\widetilde{\mathbf{x}}_{\mathrm{VI}}$ in the OP problem equals $\widetilde{\mathbf{x}}_{\mathrm{VI}}^{*}$, and 2) $\widetilde{\mathbf{x}}_{\mathrm{Line}}$ in the OP problem equals zero, such a point is a solution of the OP problem.
\end{prop}

The proof of Proposition 2 is given in Appendix A. Once the EP problem is solved, $\widetilde{\mathbf{\Omega}}$ in the domain of $\left(\widetilde{\mathbf{u}}_{\mathrm{PCC}}, \widetilde{\mathbf{u}}\right)$ can be obtained. Furthermore, $\widetilde{\mathbf{\Omega}}$ intersecting with $\mathbf{u}_{\min }-\mathbf{u}_{0} \leqslant \widetilde{\mathbf{u}} \leqslant \mathbf{u}_{\max }-\mathbf{u}_{0}$ will be projected onto the space of $\widetilde{\mathbf{u}}_{\mathrm{PCC}}$ based on any of the current methods \cite{6}-\cite{12}. By taking $\widetilde{\mathbf{u}}_{\mathrm{PCC}}=\left(\mathbf{u}_{\mathrm{PCC}}-\mathbf{u}_{\mathrm{PCC} 0}\right)$, we can get the linear constraints to characterize a sub-region $\boldsymbol{\Omega}_{\mathrm{SR}}$ of the true AC-feasible power transfer region in the domain of $\mathbf{u}_{\mathrm{PCC}}$. Based on this found sub-region, an exploration strategy will be further discussed in Sec. III-B to find more AC-feasible sub-regions.

\begin{rem}
For convenience of discussion, we assign the reference directions of branches (which may be different from the directions of actual flows) in such a way that branch $i \rightarrow j$ does not share its start node  $i$ with other branches. In this way, $\widetilde{V}_{i}$ of node $i$ can also be indexed by its incident branch index $i \rightarrow j$. Suppose $\widetilde{I}_{i j}$ is the $\dot{l}^{t h}$ element in $\widetilde{\mathbf{I}}$. This leads to a decoupled version of the EP problem across the network branches:

\noindent \textbf{EP-\textit{i}}
\begin{equation}\label{EP-i-OBJ}
\min \left\{\begin{array}{l}
\left(\widetilde{V}_{\min , i}+\widetilde{I}_{\min , i}-\widetilde{V}_{\max , i}-\widetilde{I}_{\max , i}\right) \\
+\left(\mathrm{M}_{\widetilde{V}, i}^{+}+\mathrm{M}_{\widetilde{I}, i}^{+}\right)\left(\widetilde{R}_{\max , i}^{(\mathrm{VI})}-\widetilde{R}_{\min , i}^{(\mathrm{VI})}\right)
\end{array}\right\} / H_{i},
\end{equation}
\begin{equation}\label{EP-i-C1}
\text{s.t.} \left\{\begin{array}{l}
V_{i}^{\min }-V_{0, i} \leq \widetilde{V}_{\min , i} \leq \widetilde{V}_{\max } \leq V_{i}^{\max }-V_{0, i} \\
I_{i}^{\min }-I_{0, i} \leq \widetilde{I}_{\min , i} \leq \widetilde{I}_{\max } \leq \widetilde{I}_{i}^{\max }-I_{0, i}
\end{array}\right.,
\end{equation}
\begin{equation}\label{EP-i-C2}
\left\{\begin{array}{l}
\widetilde{R}_{\min , i}^{(\mathrm{VI})} \leq \widetilde{V}_{\min , i} \widetilde{I}_{\min , i}, \widetilde{R}_{\min , i}^{(\mathrm{VI})} \leq \widetilde{V}_{\max , i} \widetilde{I}_{\min , i} \\
\widetilde{R}_{\min , i}^{(\mathrm{VI})} \leq \widetilde{V}_{\min , i} \widetilde{I}_{\max , i}, \widetilde{R}_{\min , i}^{(\mathrm{VI})} \leq \widetilde{V}_{\max , i} \widetilde{I}_{\max , i}
\end{array}\right.,
\end{equation}
\begin{equation}\label{EP-i-C3}
\left\{\begin{array}{l}
\widetilde{R}_{\max , i}^{(\mathrm{VI})} \geq \widetilde{V}_{\min , i} \widetilde{I}_{\min , i}, \widetilde{R}_{\max , i}^{(\mathrm{VI})} \geq \widetilde{V}_{\max , i} \widetilde{I}_{\min , i} \\
\widetilde{R}_{\max , i}^{(\mathrm{VI})} \geq \widetilde{V}_{\min , i} \widetilde{I}_{\max , i}, \widetilde{R}_{\max , i}^{(\mathrm{VI})} \geq \widetilde{V}_{\max , i} \widetilde{I}_{\max , i}
\end{array}\right.,
\end{equation}
over $\widetilde{V}_{\min , i}, \widetilde{V}_{\max , i}, \tilde{I}_{\min , i}, \tilde{I}_{\max , i}, \widetilde{R}_{\min , i}^{(\mathrm{VI})}, \widetilde{R}_{\max , i}^{(\mathrm{VI})}$, where $\mathrm{M}_{\widetilde{V}, i}^{+}$ is the summation of all the elements from the first $n_{N}$ rows of the $i^{t h}$ column of $\mathbf{M}_{\tilde{\mathbf{x}}}^{+}$; $\mathbf{M}_{\tilde{I}, i}^{+}$ is the summation of all the elements from the last $n_{L}$ rows of the $i^{t h}$ column of $\mathbf{M}_{\tilde{\mathbf{x}}}^{+}$. 

After solving the EP-\textit{i} problem for all the branches (or their corresponding start nodes), the solution of the EP problem is also obtained. Particularly, the minimum objective value of the EP-\textit{i} problem is non-positive (which implies a nonnegative maximum distance between the hyperplanes). This arises from the fact that zero is always a feasible solution to the EP-i problem. This decoupled treatment can be readily extended to the case where a branch $i \rightarrow j$ may share its start node $i$ with other branches.
\end{rem}

\subsection{Exploration strategy based on the found sub-regions}
Although the sub-region obtained in Sec. III-A can guarantee the AC-feasibility, it is difficult for us to know whether the obtained sub-region is overly conservative compared with the true AC-feasible power transfer region. To address this concern, we discuss an exploration strategy to find more sub-regions for a better approximation to the true AC-feasible power transfer region. Our main idea is that a vertex of $\boldsymbol{\Omega}_{\mathrm{SR}}$ can be selected as a new search point to repeat the calculation in Sec. III-A if it satisfies any of the following three conditions.

\textit{Situation 1.} The vertex only lies on the boundary of a found sub-region $\boldsymbol{\Omega}_{\mathrm{SR}}$.

\textit{Situation 2.} The number of the found sub-regions $\boldsymbol{\Omega}_{\mathrm{SR}}$ is smaller than the pre-set maximum number.

\textit{Situation 3.} The searched points do not exceed the pre-set maximum number.

Situation 1 implies that each vertex is expected to construct a new sub-region $\boldsymbol{\Omega}_{\mathrm{SR}}$ that has not been found. Situations 2 and 3 provide opportunities to terminate our characterization method at any time. When no new vertices satisfy any of the three conditions above, our characterization method will terminate. The flowchart of our characterization method is demonstrated in Fig. \ref{Fig.2}. We also propose to use different start points to initialize our characterization method. This provides an opportunity for a better approximation by getting more information from different locations in the true AC-feasible power transfer region. In this paper, such start points are selected based on the maximum and minimum power transfers, since they can provide the information for limitations of VPP adjustment capacities.
\begin{figure}[!t]
\centering
\includegraphics[width=3.2in]{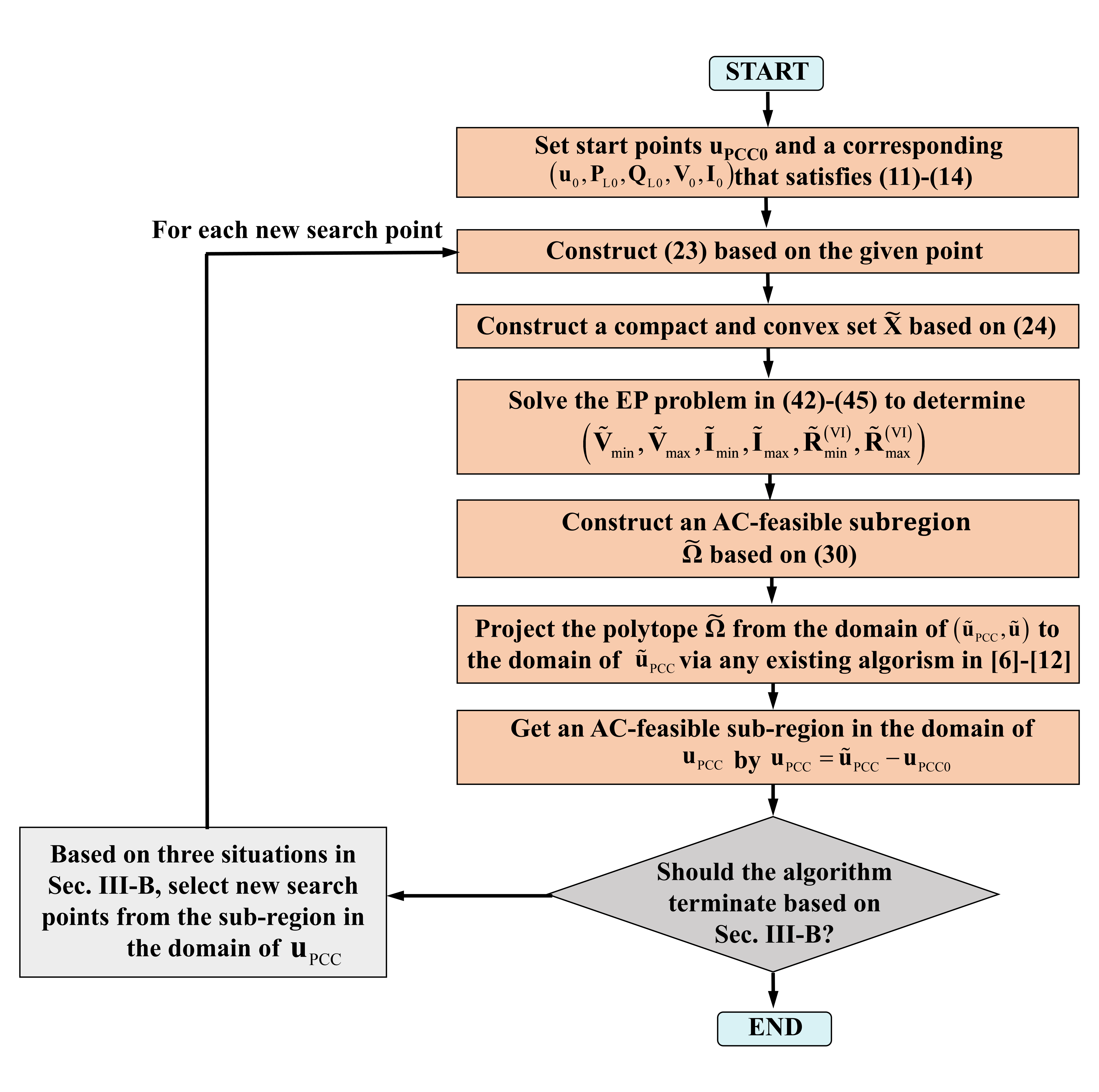}
\caption{Flowchart of the proposed characterization method.}
\label{Fig.2}
\end{figure}

\section{Application of AC-feasible power transfer regions in transmission-level operations}
Once the characterization method in Sec. V is implemented, the AC-feasible power transfer region of a VPP can be represented by the union of several sub-regions in the domain of $\mathbf{u}_\mathrm{P C C}$. The linear constraints to describe these sub-regions can be submitted to the \textit{transmission system operator} (TSO) as constraints to solve for its dispatch. In addition to the linear constraints to describe sub-regions, a VPP can also bid its PCC power transfers in a transmission electricity market [9]. This motivates the formulation of a cost function over such sub-regions in the domain of $\mathbf{u}_\mathrm{P C C}$. Ref. \cite{31} provides a general idea to formulate the equivalent cost function in the domain of $\mathbf{u}_\mathrm{P C C}$ as a convex piecewise surface. In this section, we discuss how to apply the sub-regions with their cost functions.

For convenience of discussion, we consider the coordination of one VPP with the transmission operator, while its generalization to multiple VPPs is straightforward. Assume $n_{\mathrm{J}}$ sub-regions are found based on our characterization method. Particularly, the half-space representation of the $j^{\text {th }}$ sub-region $\boldsymbol{\Omega}_{\mathrm{SR}, j}$ is denoted as $\mathbf{A}_{j} \mathbf{u}_{\mathrm{PCC}} \leqslant \mathbf{B}_{j}$. Also, the equivalent cost function over $\boldsymbol{\Omega}_{\mathrm{SR}, j}$ is a piecewise convex function denoted by $z_{j}\left(\mathbf{u}_{\mathrm{PCC}}\right)$ obtained from \cite{31}.

The incorporation of our sub-regions and their equivalent cost functions into the transmission-level operation can be achieved by introducing binary variables to determine a sub-region in which the final PCC power is located. Based on this idea, we formulate the following optimization problem for the transmission-level operation:

\noindent \textbf{TP:}
\begin{equation}\label{TP-OBJ}
\min C(\mathbf{s})+\sum_{j=1}^{n_{\mathrm{J}}}\left(k_{j} z_{j}\left(\mathbf{u}_{\mathrm{PCC}}\right)\right)
\end{equation}
\begin{equation}\label{TP-C1}
\text{s.t.} \mathbf{g}_{\mathrm{te}}\left(\mathbf{s}, \mathbf{u}_{\mathrm{PCC}}\right)=\mathbf{0},
\end{equation}
\begin{equation}\label{TP-C2}
\mathbf{g}_{\mathrm{tie}}\left(\mathbf{s}, \mathbf{u}_{\mathrm{PCC}}\right) \leq \mathbf{0},
\end{equation}
\begin{equation}\label{TP-C3}
k_{j} \mathbf{A}_{j} \mathbf{u}_{\mathrm{PCC}} \leq k_{j} \mathbf{B}_{j}, \forall j,
\end{equation}
\begin{equation}\label{TP-C4}
\sum_{j=1}^{n_\mathrm{J}} k_{j}=1,
\end{equation}
\begin{equation}\label{TP-C5}
k_{j}=\{0,1\}, \forall j,
\end{equation}
over $\mathbf{s}, \mathbf{u}_{\mathrm{PCC}}, k_{j}, \forall j$, where $\mathbf{s} \in \mathbb{R}^{n_{\mathrm{s}} \times 1}$ is the dispatch variables of the transmission operator; $C(\mathbf{s})$ is a convex cost function of the transmission operator to dispatch $\mathbf{S}$; $\mathbf{g}_\mathrm{te}: \mathbb{R}^{\left(n_{\mathrm{PCC}}+n_{\mathrm{s}}\right) \times 1} \rightarrow \mathbb{R}^{n_{\mathrm{tie}} \times 1}$; $\mathbf{g}_\mathrm{tie}: \mathbb{R}^{\left(n_{\mathrm{PCC}}+n_{\mathrm{s}}\right) \times 1} \rightarrow \mathbb{R}^{n_{\mathrm{tie}} \times 1}$, where $n_{\text {te}}$ and $n_{\text {tie}}$ are respectively the numbers of equalities and inequalities in the transmission network.

The objective function (\ref{TP-OBJ}) is to minimize the total cost on the transmission and VPP sides. The constraint (\ref{TP-C1}) is the compact formulation of transmission-level requirements that can be expressed as equalities, such as power balance. The constraint (\ref{TP-C2}) is the compact formulation of transmission-level requirements that can be expressed as inequalities, such as flow limits. The constraint (\ref{TP-C3}) describes the $j^{t h}$ sub-region in our characterization method. The constraints (\ref{TP-C4})-(\ref{TP-C5}) guarantees the unique selection of such a sub-region. In power industries, the DC model is usually employed by the transmission operator to formulate $\mathbf{g}_\mathrm{te}$ and $\mathbf{g}_\mathrm{tie}$ [27]. However, even if the DC model is employed, the TP problem is still a MINLP which is difficult to solve. Consequently, a big-M formulation is proposed below to linearize the TP problem.

\noindent \textbf{E-TP:}
\begin{equation}\label{E-TP-OBJ}
\min C(\mathbf{s})+\sum_{j=1}^{n_\mathrm{J}} y_{j},
\end{equation}
\begin{equation}\label{E-TP-C1}
\text{s.t. Constraints (\ref{TP-C1})-(\ref{TP-C2}) and (\ref{TP-C4})-(\ref{TP-C5})},
\end{equation}
\begin{equation}\label{E-TP-C2}
\mathbf{A}_{j} \mathbf{u}_{\mathrm{PCC}} \leq \mathbf{B}_{j}+\left(1-k_{j}\right) \mathrm{M}, \forall j,
\end{equation}
\begin{equation}\label{E-TP-C3}
z_{j}\left(\mathbf{u}_{\mathrm{PCC}}\right) \leq y_{j}+\left(1-k_{j}\right) \mathrm{M}, \forall j,
\end{equation}
\begin{equation}\label{E-TP-C4}
0 \leq y_{j}, \forall j,
\end{equation}
over $\mathbf{S}, \mathbf{u}_\mathrm{P C C}, k_{j}, y_{j}, \forall j$,
\noindent where $y_{j}$ is is an ancillary variable in real space; M is a real number that is positive and sufficiently large.

The E-TP problem is a \textit{mixed-integer linear program} (MILP) instead of a MINLP. It is not hard to justify that the E-TP problem is equivalent to the TP problem, i.e., the dispatch decision solved from the E-TP problem is the same as that from the TP problem. Particularly, $\sum_{j=1}^{n_{\mathrm{J}}} y_{j}$ represents the cost of the VPP in the transmission-level operation.

\section{Case Studies}
The proposed methods are verified in the IEEE 33-bus and IEEE 136-bus test systems. System parameters can be found in \cite{32}. All numerical results are calculated with MATLAB R2012a and performed on a laptop equipped with Intel (R) Core (TM) i7-8565U CPU @ 1.80GHz 8.00G RAM. The modeling of all optimization problems is via YALMIP. All NLPs are solved via IPOPT, while all mixed-integer programs are solved under the default settings in YALMIP.

\subsection{Validation of the characterization method }
In this sub-section, the following three methods will be employed in the IEEE 33-bus test system to calculate the AC-feasible power transfer region:
\begin{figure}[b!]
\centering
\includegraphics[width=3.2in]{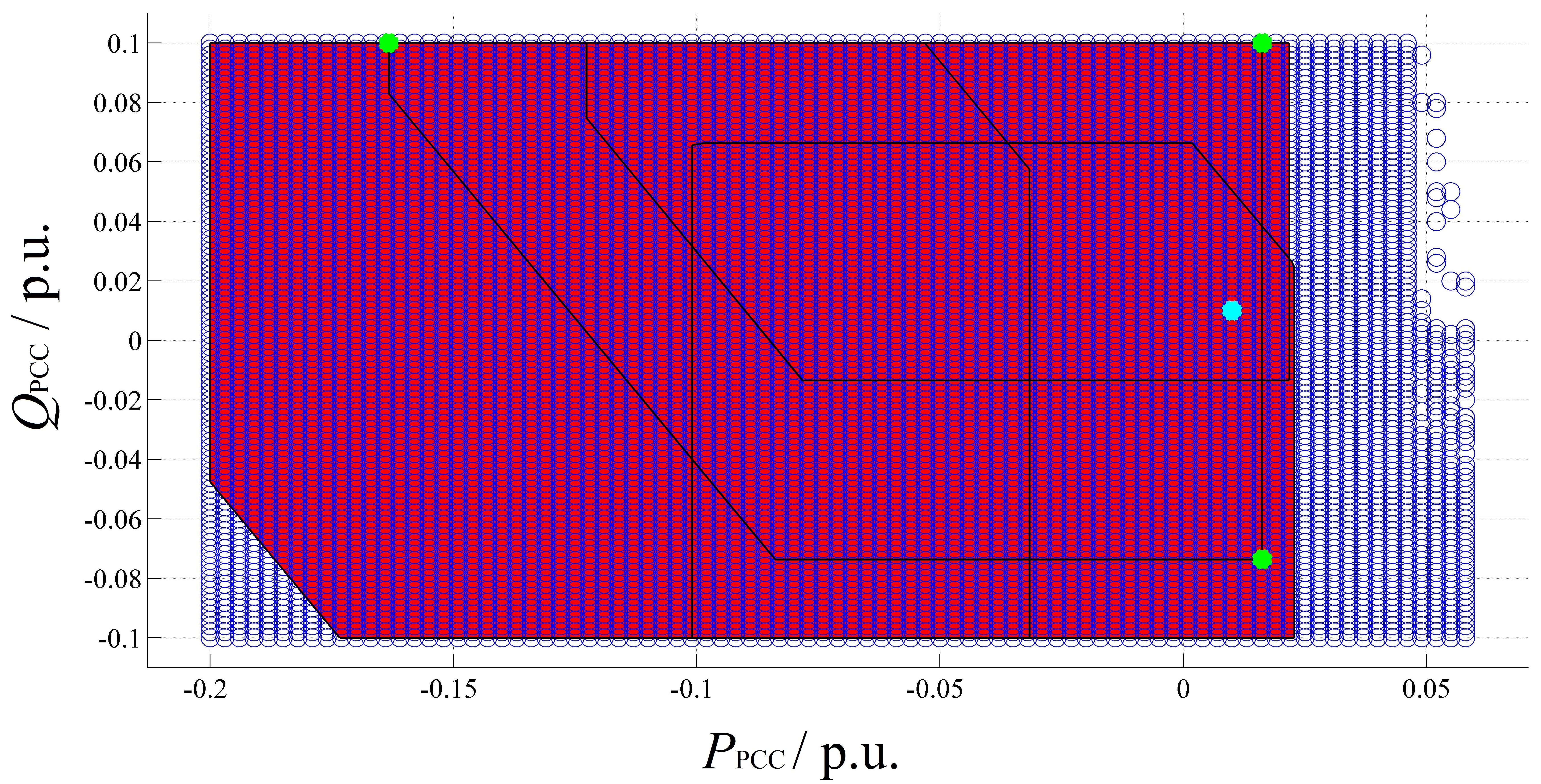}
\caption{Comparison of methods M0-M2 in the IEEE 33-bus test system.}
\label{Fig.3}
\end{figure}

\textbf{M0}: The AC-feasible power transfer region is plotted by brute-force search of discrete AC-feasible points. The searched range of $\mathbf{u}_{\mathrm{PCC}}=\left(P_{\mathrm{PCC}}, Q_{\mathrm{PCC}}\right)$ is evenly discretized as 101×101=10201 points. Given each point, its feasibility for the constraints (\ref{C1})-(\ref{C4}) is checked. The feasible region obtained by this method is regarded as a benchmark.

\textbf{M1}: Our characterization method proposed in Sec. III.

\textbf{M2}: The method in \cite{20} which is a representative to guarantee the AC-feasibility by linear constraints. Particularly, the feasible region is pre-defined as a rectangle in \cite{20}.

In the IEEE 33-bus test system, the M1 method (i.e., the proposed method) stops after the second iteration. As demonstrated in Fig. \ref{Fig.3}, four points (the cyan point for the first iteration and the three green points for the second iteration) are used to construct four sub-regions in which AC-feasibility is guaranteed. The union of the four sub-regions serves as the result of our method.

The comparison between methods M0-M2 is also shown in Fig. \ref{Fig.3}. The true AC-feasible power transfer region (i.e., blue points) obtained by the M0 method is non-convex. The proposed M1 method provides an inner approximation to the true feasible region with four sub-regions. Particularly, 86.18 $\%$ of feasible points in the M0 method are contained in the result of two iterations in the M1 method. Meanwhile, given any feasible point as a start point, the previous M2 method can also construct a feasible rectangular inner approximation, as reported in \cite{20}. Particularly, we use the same set of four points as in our proposed method to serve as the start points for the M2 method. In this test system, the M2 method just returns the four points themselves as the result inner approximation, without any further exploration. Our conjecture is that the restrictive rectangular shape in the M2 method makes the inner approximation so conservative in this test system.

\begin{figure}[b!]
\centering
\includegraphics[width=3.2in]{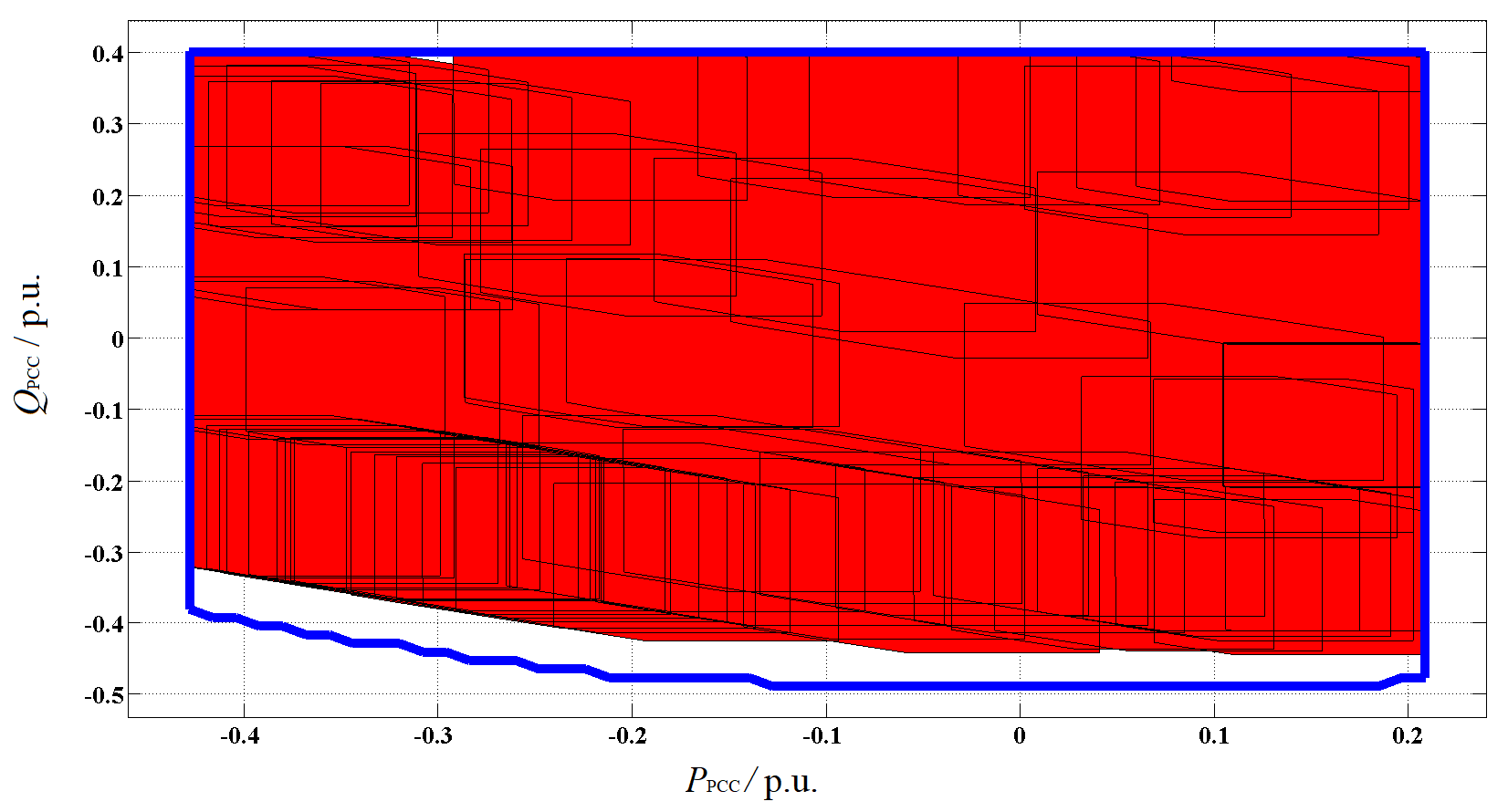}
\caption{Comparison of our characterization with the true feasible region in the IEEE 136-bus test system. The blue solid lines are the boundaries of the true AC-feasible power transfer region. A red shape indicates a sub-region in our characterization method.}
\label{Fig.4}
\end{figure}

We further test our characterization method in the IEEE 136-bus test system. In this test system, we set up four different points to initialize our characterization method in the first iteration. These start points are selected based on the maximum and minimum power transfers subject to the original constraints (\ref{C1})-(\ref{C4}) to reflect the limitations of VPP adjustment capacities. 70 sub-regions are constructed based on our characterization method, and their union is compared with the boundaries of the true AC-feasible power transfer region obtained by brute-force search, as shown in Fig. \ref{Fig.4}. In Fig. \ref{Fig.4}, our result covers 94.66$\%$ of the true AC-feasible power transfer region obtained by brute-force search. This indicates that the union of 70 sub-regions has a good approximate to the true AC-feasible power transfer region. Particularly, the union of any number of such sub-regions can serve as an inner approximation to the true AC-feasible power transfer region. Consequently, when we stop our characterization method at any time, our result guarantees the valid linear constraints to satisfy AC-feasibility.

When we construct sub-regions shown in Figs. \ref{Fig.3}-\ref{Fig.4}, our characterization method needs to solve the EP problem which is nonconvex. Its computational concern can be alleviated by the following features of our method.

First, the EP problem can be decomposed as EP-\textit{i} problems by branches, as demonstrated in Remark 1. Particularly, for branch $i \rightarrow j$ whose start node $i$ is not shared with other branches, its EP-\textit{i} problem only has six variables $\left(\widetilde{V}_{\min , i}, \widetilde{V}_{\max , i}, \widetilde{I}_{\min , i}, \widetilde{I}_{\max , i}, \widetilde{R}_{\min , i}^{(\mathrm{VI})}, \widetilde{R}_{\max , i}^{(\mathrm{VI})}\right)$ which subject to the constraints (\ref{EP-C1})-(\ref{EP-C3}). Consequently, scale of the EP-\textit{i} problem is small. This decoupled structure of the EP problem facilitates its solution.

Second, our method is adaptive to the changing conditions in a distribution network. The result of our method in the $\left(\mathbf{u}_{\mathrm{PCC}}, \mathbf{u}\right)$ domain remains the same as long as the topology and parameters of the distribution network remain unchanged. The feasibility under any load changes can just be immediately checked by the result in $\left(\mathbf{u}_{\mathrm{PCC}}, \mathbf{u}\right)$. Moreover, we can perform the proposed method offline with respect to a set of typical topologies and parameters of the network, to save time during online implementation.

\subsection{Validation of the transmission-level application}
In this sub-section, the application method in Sec. IV to coordinate VPPs and the transmission network will be verified in the IEEE 136-bus test system. Ten IEEE 136-bus test systems are coordinated with a 661-bus transmission system from a province in China. Sub-regions in Fig. \ref{Fig.4} and their equivalent cost functions obtained by [31] will be utilized. As for the 661-bus transmission system, the DC power flow model in [32] is employed following the common practice of TSOs. Two problems are solved for the coordination between the ten VPPs and the TSO system: one is the TP problem (a MINLP), and the other is the E-TP problem (a MILP) based on our big-M linearization. In this testing, 1438 continuous variables, 700 binary variables, and 23692 constraints are involved in the E-TP problem. The scale of the TP problem is similar. It is found that solving the TP problem requires more than half an hour, which exceeds the common time window allowed (e.g., 15 minutes in Guangdong Power Grid, China). In contrast, our proposed method can solve the E-TP problem in 539.6 seconds, while satisfying the AC feasibility of PCC power exchanges.

\section{Conclusions}
To facilitate the participation of a \textit{virtual power plant} (VPP) in the transmission-level operation, a characterization method is proposed in this paper to determine the feasible power transfer region of a VPP. The AC feasibility of such a region can be guaranteed by the proposed characterization method. We also developed a big-M formulation to accurately linearize the coordinated operation problem for the VPPs and the transmission operator. Numerical results in the IEEE 33-bus and the IEEE 136-bus test systems demonstrate the effectiveness of the proposed methods.

\section*{Appendix A: Proof of Proposition 2}
Consider an optimization problem R-OP, which is a relaxation to the OP problem by removing the constraint (\ref{P1-1}). The R-OP problem can be decomposed as two optimization problems whose variables and constraints are respectively associated with $\widetilde{\mathbf{x}}_{\mathrm{VI}}$ and $\widetilde{\mathbf{x}}_{\mathrm {Line }}$ defined in (\ref{Solution}):

\noindent \textbf{R-OP-1:}
\begin{small}
\begin{equation}\label{A1}
\min \mathbf{e}_{1}^{T}\left\{\left\{\mathbf{M}_{\tilde{\mathbf{x}}}^{+}\left(\left(\tilde{\mathbf{R}}_{\max }^{(\mathrm{P})}-\tilde{\mathbf{R}}_{\min }^{(\mathrm{P})}\right)+\left(\tilde{\mathbf{R}}_{\max }^{(\mathrm{Q})}-\tilde{\mathbf{R}}_{\min }^{(\mathrm{Q})}\right)\right)\right\} \cdot / \mathbf{H}\right\},
\end{equation}
\end{small}
\begin{equation}\label{A2}
\text{s.t. Constraints (\ref{RmaxP})-(\ref{RminQ})},
\end{equation}
\noindent over $\widetilde{\mathbf{P}}_{\text {Lmin }}$,$\widetilde{\mathbf{P}}_{\text {Lmax }}$,$\widetilde{\mathbf{Q}}_{\text {Emin }}$,$\widetilde{\mathbf{Q}}_{\text {Emax }}$,$\widetilde{\mathbf{R}}_{\min }^{(\mathrm{P})}$,$\widetilde{\mathbf{R}}_{\max }^{(\mathrm{P})}$,$\widetilde{\mathbf{R}}_{\min }^{(\mathrm{Q})}$,$\widetilde{\mathbf{R}}_{\max }^{(\mathrm{Q})}$.

\noindent \textbf{R-OP-2:}
\begin{small}
\begin{equation}\label{A3}
\min \mathbf{e}_{1}^{T}\left\{\left\{\begin{array}{l}
{\left[\begin{array}{ll}
\widetilde{\mathbf{V}}_{\min }^{T} & \widetilde{\mathbf{I}}_{\min }^{T}
\end{array}\right]^{T}-\left[\begin{array}{ll}
\widetilde{\mathbf{V}}_{\max }^{T} & \widetilde{\mathbf{I}}_{\max }^{T}
\end{array}\right]^{T}} \\
+\mathbf{M}_{\widetilde{\mathbf{x}}}^{+}\widetilde{\mathbf{R}}_{\max }^{(\mathrm{VI})}-\mathbf{M}_{\widetilde{\mathbf{x}}}^{+}\widetilde{\mathbf{R}}_{\min }^{(\mathrm{VI})}
\end{array}\right\} . / \mathbf{H}\right\},
\end{equation}
\end{small}
\begin{equation}\label{A4}
\text{s.t.} \quad \text{Constraints (\ref{X2}) and (\ref{EP-C1})-(\ref{EP-C3})},
\end{equation}
\noindent over $\widetilde{\mathbf{V}}_{\min }, \widetilde{\mathbf{V}}_{\max }, \widetilde{\mathbf{I}}_{\min }, \widetilde{\mathbf{I}}_{\max }, \widetilde{\mathbf{R}}_{\min }^{(\mathrm{VI})}, \widetilde{\mathbf{R}}_{\max }^{(\mathrm{VI})}$.

For the R-OP-1 problem, the terms $\left(\widetilde{\mathbf{R}}_{\max }^{(\mathrm{P})}-\widetilde{\mathbf{R}}_{\min }^{(\mathrm{P})}\right)$ and $\left(\widetilde{\mathbf{R}}_{\max }^{(\mathrm{Q})}-\widetilde{\mathbf{R}}_{\min }^{(\mathrm{Q})}\right)$ in the objective function (\ref{A1}) are non-negative. Consequently, the optimal objective value of the R-OP-1 problem is zero, which occurs when $\widetilde{\mathbf{x}}_{\mathrm {Line }}$ attains zero.

The R-OP-2 problem is equivalent to the EP problem by applying Lemma 1 in Appendix B. We will next show that the constraint (\ref{P1-1}) is also fulfilled by the optimal solution of the R-OP-2 problem (i.e., the EP problem) and $\widetilde{\mathbf{x}}_{\mathrm {Line }}=\mathbf{0}$. When $\widetilde{\mathbf{x}}_{\mathrm {Line }}=0$, the constraint (\ref{P1-1}) becomes
\begin{equation}\label{A5}
\begin{aligned}
&{\left[\begin{array}{ll}
\tilde{\mathbf{V}}_{\min }^{T} & \tilde{\mathbf{I}}_{\min }^{T}
\end{array}\right]^{T}-\left[\begin{array}{ll}
\tilde{\mathbf{V}}_{\max }^{T} & \tilde{\mathbf{I}}_{\max }^{T}
\end{array}\right]^{T}} \\
&\leq \mathbf{M}_{\tilde{\mathbf{x}}}^{+} \tilde{\mathbf{R}}_{\min }^{(\mathrm{VI})}-\mathbf{M}_{\tilde{\mathbf{x}}}^{+} \widetilde{\mathbf{R}}_{\max }^{(\mathrm{VI})}
\end{aligned}.
\end{equation}

Given the optimal solution $\widetilde{\mathbf{x}}_{\mathrm{VI}}^{*}=\left(\widetilde{\mathbf{V}}_{\min }^{*}, \widetilde{\mathbf{V}}_{\max }^{*}, \widetilde{\mathbf{I}}_{\min }^{*}, \widetilde{\mathbf{I}}_{\max }^{*}, \widetilde{\mathbf{R}}_{\min }^{(\mathrm{VI})^{*}}, \widetilde{\mathbf{R}}_{\max }^{(\mathrm{VI}) *}\right)$ of the EP problem, the fulfillment of the constraint (\ref{P1-1}) can be checked by solving the following LP:

\noindent \textbf{P-LP:}
\begin{equation}\label{A6}
\min \mathbf{e}_{1}^{T} \mathbf{s}_{1}
\end{equation}
\begin{equation}\label{A7}
\text{s.t.} \quad \mathbf{s}_{1} \geq \mathbf{0},
\end{equation}
\begin{equation}\label{A8}
\begin{aligned}
&{\left[\left(\widetilde{\mathbf{V}}_{\min }^{*}\right)^{T}\left(\widetilde{\mathbf{I}}_{\min }^{*}\right)^{T}\right]^{T}-\left[\left(\widetilde{\mathbf{V}}_{\max }^{*}\right)^{T}\left(\widetilde{\mathbf{I}}_{\max }^{*}\right)^{T}\right]^{T}} \\
&\leq-\mathbf{M}_{\widetilde{\mathbf{x}}}^{+} \widetilde{\mathbf{R}}_{\max }^{(\mathrm{VI})^{*}}+\mathbf{M}_{\widetilde{\mathbf{x}}}^{+} \widetilde{\mathbf{R}}_{\min }^{(\mathrm{VI})^{*}}+\mathbf{s}_{1}
\end{aligned},
\end{equation}
\noindent over $\mathbf{s}_{1} \in \mathbb{R}^{\left(n_{\mathrm{L}}+n_{\mathrm{N}}\right) \times 1}$.

By strong duality, the minimum value of the P-LP problem equals the maximum objective value of its dual problem:

\noindent \textbf{D-LP:}
\begin{small}
\begin{equation}\label{A9}
\max \mathbf{\lambda}^{T}\left\{\begin{array}{ll}
{\left[\begin{array}{ll}
\left(\widetilde{\mathbf{V}}_{\text {m in }}^{*}\right)^{T} & \left(\widetilde{\mathbf{I}}_{\text {m in }}^{*}\right)^{T}
\end{array}\right]^{T}+\mathbf{M}_{\widetilde{\mathbf{x}}}^{+} \widetilde{\mathbf{R}}_{\text {max }}^{(\mathrm{VI})^{*}}} \\
-\left[\begin{array}{ll}
\left(\widetilde{\mathbf{V}}_{\max }^{*}\right)^{T} & \left(\widetilde{\mathbf{I}}_{\max }^{*}\right)^{T}
\end{array}\right]^{T}-\mathbf{M}_{\widetilde{\mathbf{x}}}^{+} \widetilde{\mathbf{R}}_{\min }^{(\mathrm{VI})^{*}}
\end{array}\right\},
\end{equation}
\end{small}
\begin{equation}\label{A10}
\text{s.t.} \quad \mathbf{0} \leq \mathbf{\lambda} \leq \mathbf{1},
\end{equation}
\noindent over $\mathbf{\lambda} \in \mathbb{R}^{\left(n_{\mathrm{L}}+n_{\mathrm{N}}\right) \times 1}$.

The optimal solution of the EP problem satisfies the constraint (\ref{A5}), if and only if the maximum objective value of the D-LP problem is zero. In the remaining part of this proof, we will show that the zero maximum dual objective value is indeed attained.

The maximum objective value of the D-LP problem is taken at a vertex of the constraint (\ref{A10}). i.e., it equals
\begin{equation}\label{A11}
\left\{z_{0}^{*}, z_{1}^{*}, \ldots, z_{\mathrm{S}-1}^{*}\right\}^{\max },
\end{equation}
\noindent where $\mathrm{S}=2^{\left(n_\mathrm{L}+n_{\mathrm{N}}\right)}$ is the number of vertices and $z_{k}^{*}$ denotes the value of 
\begin{equation}\label{A12}
\lambda_{k}^{T}\left\{\begin{array}{ll}
{\left[\left(\widetilde{\mathbf{V}}_{\min }^{*}\right)^{T}\right.} \left.\left(\widetilde{\mathbf{I}}_{\min}^{*}\right)^{T}\right]^{T}+\mathbf{M}_{\widetilde{\mathbf{x}}}^{+} \widetilde{\mathbf{R}}_{\max }^{(\mathrm{VI})^{*}} \\
-\left[\begin{array}{ll}
\left(\widetilde{\mathbf{V}}_{\max }^{*}\right)^{T} & \left(\widetilde{\mathbf{I}}_{\max }^{*}\right)^{T}
\end{array}\right]^{T}-\mathbf{M}_{\widetilde{\mathbf{x}}}^{+} \widetilde{\mathbf{R}}_{\min }^{(\mathrm{VI})^{*}}
\end{array}\right\}.
\end{equation}

Particularly, all the elements in $\mathbf{\lambda}_{0}$ are zero and thus $z_{0}^{*}$ equals zero. Meanwhile, at least one element in $\mathbf{\lambda}_{k}(k=1,2, \ldots, \mathrm{S}-1)$ equals one. Based on Remark 1, (\ref{A12}) can be separated based on branches, i.e.,
\begin{footnotesize}
\begin{equation}\label{A13}
\begin{aligned}
z_{k}^{*} &=\sum_{i=1}^{n_{L}} z_{k, i}^{*} \\
&=\sum_{i=1}^{n_{L}}\left\{\begin{array}{l}
\lambda_{k, i}\left(\widetilde{V}_{\min , i}^{*}-\widetilde{V}_{\max , i}^{*}\right)+\lambda_{k, i+n_{\mathrm{N}}}\left(\widetilde{I}_{\min , i}^{*}-\widetilde{I}_{\max , i}^{*}\right) \\
+\left(\lambda_{k, i} M_{\widetilde{V}, i}^{+}+\lambda_{k, i+n_{\mathrm{N}}} M_{\widetilde{I}, i}^{+}\right)\left(\widetilde{R}_{\mathrm{max}, i}^{(\mathrm{VI})^{*}}-\widetilde{R}_{\min , i}^{(\mathrm{VI})^{*}}\right)
\end{array}\right\}.
\end{aligned}
\end{equation}
\end{footnotesize}
By applying Lemma 2 in Appendix B, each $z_{k, i}^{*}$ in (\ref{A13}) is non-positive and thus $z_{k}^{*}$ is also non-positive. Because $z_{0}^{*}$ is zero and $z_{k}^{*}$ is non-positive, the maximum objective value of the D-LP problem is zero, i.e., the constraint (\ref{P1-1}) is fulfilled with the optimal solution of the EP problem and $\widetilde{\mathbf{x}}_{\mathrm {Line }}=\mathbf{0}$.

\section*{Appendix B: Lemmas 1 and 2}
Suppose there are 1) continuous variables $x_{\max }$, $x_{\min }$, $y_{\max }$, $y_{\min }$, $r_{\min }$, $r_{\min } \in \mathbb{R}$, 2) coefficients $\mathrm{M}_{x}^{+}$ and $\mathrm{M}_{y}^{+} \in\left\{0 \cup \mathbb{R}^{+}\right\}$, and 3) bounds $x_{\max,0 }$, $x_{\min,0 }$, $y_{\max,0 }$,$y_{\min,0 } \in \mathbb{R}$ and $x_{\min , 0} \leqslant 0 \leqslant x_{\max , 0}$ and $y_{\min , 0} \leqslant 0 \leqslant y_{\max , 0}$. Define an optimization problem as
\begin{equation}\label{B1}
\min x_{\min }+y_{\min }-x_{\max }-y_{\max }+\left(\mathrm{M}_{x}^{+}+\mathrm{M}_{y}^{+}\right)\left(r_{\max }-r_{\min }\right),
\end{equation}
\begin{equation}\label{B2}
\text{s.t.} \quad
\left\{\begin{array}{l}
x_{\min , 0} \leq x_{\min } \leq x_{\max } \leq x_{\max , 0} \\
y_{\min , 0} \leq y_{\min } \leq y_{\max } \leq y_{\max , 0}
\end{array}\right.,
\end{equation}
\begin{equation}\label{B3}
\left\{\begin{array}{l}
r_{\max } \geq x_{\min } y_{\min }, r_{\max } \geq x_{\max } y_{\max } \\
r_{\max } \geq x_{\max } y_{\min }, r_{\max } \geq x_{\min } y_{\max }
\end{array}\right.,
\end{equation}
\begin{equation}\label{B4}
\left\{\begin{array}{l}
r_{\min } \leq x_{\max } y_{\min }, r_{\min } \leq x_{\min } y_{\max } \\
r_{\min } \leq x_{\max } y_{\max }, r_{\min } \leq x_{\min } y_{\min }
\end{array}\right.,
\end{equation}
\noindent over $x_{\min }, x_{\max }, y_{\min }, y_{\max }, r_{\min }, r_{\max }$.

\begin{Lemma}
The optimal solutions $r_{\max }^{*}$ and $r_{\min }^{*}$ equal 
\begin{equation}\label{B5}
\left\{\begin{aligned}
r_{\max }^{*} &=\left\{x_{\min }^{*} y_{\min }^{*}, x_{\min }^{*} y_{\max }^{*}, x_{\max }^{*} y_{\min }^{*}, x_{\max }^{*} y_{\max }^{*}\right\}_{\max } \\
r_{\min }^{*} &=\left\{x_{\min }^{*} y_{\min }^{*}, x_{\min }^{*} y_{\max }^{*}, x_{\max }^{*} y_{\min }^{*}, x_{\max }^{*} y_{\max }^{*}\right\}_{\min }
\end{aligned}\right..
\end{equation}
\end{Lemma}

The proof of Lemma 1 is not provided in this paper since this result is straightforward.
\begin{Lemma}
Given $\lambda_{x}$ and $\lambda_{y} \in\{0,1\}$ and the optimal solutions (\ref{B5}), the following inequality holds
\begin{equation}\label{B6}
\begin{aligned}
&\lambda_{x}\left(x_{\min }^{*}-x_{\max }^{*}\right)+\lambda_{y}\left(y_{\min }^{*}-y_{\max }^{*}\right) \\
&\leq\left(\lambda_{x} \mathrm{M}_{x}^{+}+\lambda_{y} \mathrm{M}_{y}^{+}\right)\left(r_{\min }^{*}-r_{\max }^{*}\right)
\end{aligned}.
\end{equation}
\end{Lemma}
\begin{proof}
$\left(\lambda_{x}, \lambda_{y}\right)$ can be selected as (0,0), (1,1), (1,0), and (0,1). Obviously, when $\left(\lambda_{x}, \lambda_{y}\right)=(0,0)$ and $\left(\lambda_{x}, \lambda_{y}\right)=(1,1)$, (\ref{B6}) holds. When $\left(\lambda_{x}, \lambda_{y}\right)=(1,0)$, (\ref{B6}) becomes
\begin{equation}\label{B7}
\left(x_{\min }^{*}-x_{\max }^{*}\right)+\mathrm{M}_{x}^{+}\left(r_{\max }^{*}-r_{\min }^{*}\right) \leq 0.
\end{equation}
The optimization problem (\ref{B1})-(\ref{B4}) can be regarded by iteratively solving variables in sequence. The variables (ymin, ymax) are firstly fixed and then variables $\left(y_{\min }, y_{\max }\right)$ are are firstly fixed and then variables $\left(x_{\min }, x_{\max }, r_{\min }, r_{\max }\right)$ are determined. Under such a case, (\ref{B7}) is smaller than zero at optimum because $\left(x_{\min }, x_{\max }, r_{\min }, r_{\max }\right)$ can attain zero and the optimization problem (\ref{B1})-(\ref{B4}) is a minimization problem. Similarity, the fulfillment of (\ref{B6}) can be proven when $\left(\lambda_{x}, \lambda_{y}\right)=(0,1)$.
\end{proof}

\end{document}